%% file: fermionic_error_correction.tex
\documentclass[letterpaper,reprint,amsmath,amssymb,aps,pra,longbibliography,nofootinbib,floatfix]{revtex4-2}
\pdfoutput=1    

\input{header}

\begin{document}

\preprint{APS/123-QED}

\title{Majorana qubit codes that also correct odd-weight errors}

\author{Sourav Kundu}
\email{souravku@usc.edu}
\author{Ben W. Reichardt}
\affiliation{Ming Hsieh Department of Electrical Engineering, University of Southern California, Los Angeles, California 90007, USA}


\begin{abstract}
The tetron architecture is a promising candidate for topological quantum computation. Each tetron Majorana island has four Majorana zero modes, and possible measurements are constrained to span zero or two Majoranas per tetron. Such measurements are known to be sufficient for correcting so-called ``bosonic errors," which affect an even number of Majoranas per tetron. We demonstrate that such measurements are also sufficient for correcting ``fermionic errors," which affect an odd number of Majoranas per tetron. In contrast, previous proposals for ``fermionic error correction" on tetrons introduce more experimental challenges. 

We show that ``fermionic codes" can be derived from traditional ``bosonic codes" by inclusion of tetrons in the stabilizer group. 
  
\end{abstract}

\maketitle

\section{Introduction} 
\label{sec:introduction}

Topological quantum computation has garnered an increasing interest in the recent years \cite{Kitaev2001, Kitaev2003, Kitaev2006, Nayak2008, Sau2010, Bravyi2010, Alicea2012, Leijnse2012, Mong2014, Sarma2015, Aasen2016, Lutchyn2018, Flensberg2021, Aghaee2023}. 
Just as topology cannot be changed by local perturbations, likewise topologically protected quantum information cannot be affected by most local errors. Some local errors can be problematic and need correction. In particular, we discuss error correction schemes for the Majorana-based tetron architecture \cite{Karzig2017, Litinski2017, Litinski2018, Knapp2018, Karzig2019}, shown in \figref{fig:tetron_hardware}.

\begin{figure}[htpb]
	{\includegraphics[width=0.843\linewidth]{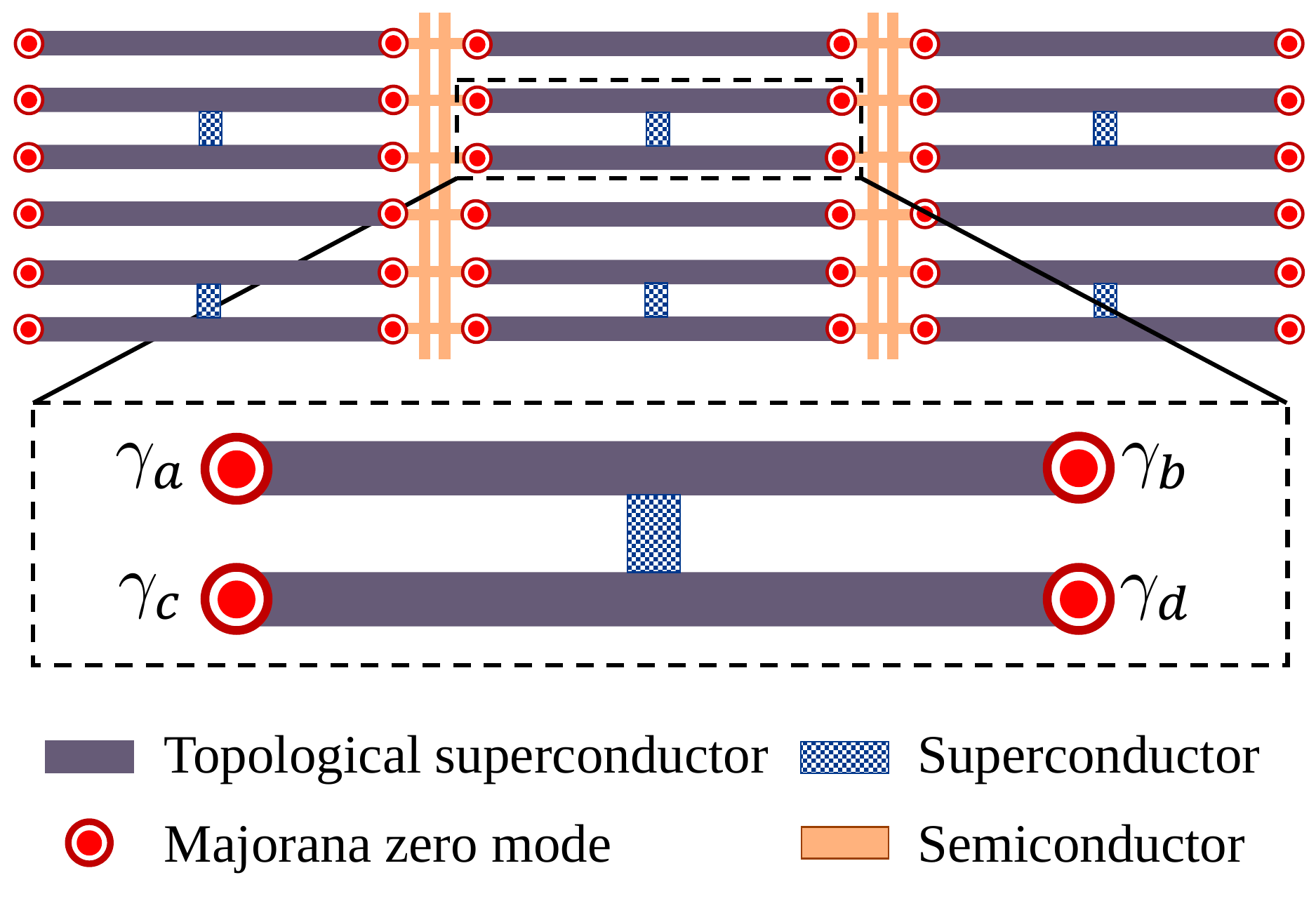}}
	\caption{
		The above figure illustrates a tetron architecture. A single tetron comprises two topological superconducting nanowires that host four Majorana zero modes at their ends, and they are connected by a superconductor.}
	\label{fig:tetron_hardware}
\end{figure}

A tetron is a superconducting island, hosting four localized Majorana zero modes (MZMs), schematically shown in \figref{fig:tetron}. One can measure the fermion parity of any operator that spans zero or two MZMs per tetron. Using such measurements, we can define stabilizer codes in a system of several tetrons. 
The possible errors in a tetron architecture include ``bosonic" and ``fermionic" errors. Bosonic errors are those which affect two MZMs per tetron, while fermionic errors affect an odd number of MZMs per tetron. Bosonic errors can be mapped to Pauli errors, and can be corrected using conventional stabilizer codes. This is not the case for fermionic errors.

\begin{figure}[t]
        \begin{tabular}{c c}
        		\hspace{5 pt}
                \subfigure[\label{fig:tetron}]{\includegraphics[width=0.15\linewidth]{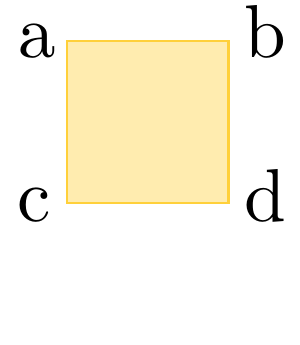}}
                \hspace{12 pt}
                &
                \subfigure[\label{fig:tetron_representation}]{\includegraphics[width=0.7\linewidth]{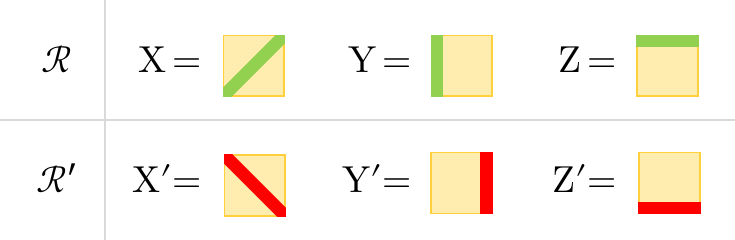}}
        \end{tabular}
        \caption{
        		(a) A tetron hosts four MZMs at the locations \(a,b,c,d\). 
                (b) \(\mathcal{R}\) and \(\mathcal{R}'\) are two complementary sets of weight-2 operators in a tetron.} 
        \label{fig:tetron_and_representation}
\end{figure}

Fermionic errors are typically suppressed by high charging energy. Short-lived fermionic errors relax to bosonic errors. However, a sufficiently long-lived fermionic error can disrupt measurement outcomes, and even spread to adjacent tetrons during connected measurements. Hence Ref.~\onlinecite{Knapp2018} suggests the use of ``fermionic codes'' that can correct fermionic errors, so that they do not spread and cause a logical error.

Although several Majorana fermionic codes have been proposed \cite{Hastings2017a,Vijay2017}, there exist implementation challenges specific to the tetron architecture. For example, one of the chief experimental hurdles in employing a fermionic code is the inability to directly measure the four-MZM parity of a tetron \cite{Knapp2018}. If the tetron parity were measurable, then a fermionic error could be detected. However, a measurement can overlap only two MZMs per tetron. We will show that such measurements suffice for fermionic error correction with two strategies.

In Sec.~\ref{sec:fermionic_codes_1}, we derive fermionic error correcting codes from conventional bosonic codes such as color codes and surface codes. We introduce additional stabilizers in these codes, in which one qubit has a different Pauli to MZM operator mapping. This places the tetron in the stabilizer group, enabling fermionic error correction. \Figref{fig:example_1_recipe_1_and_pictorial_product} shows an example.

\begin{figure}[htpb]
        \begin{tabular}{c}
        		\subfigure[\label{fig:example_1_generators}]{\raisebox{0mm}{\includegraphics[width=0.4\linewidth]{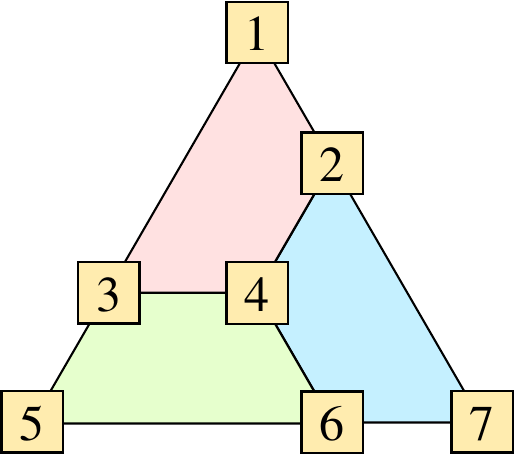}}}
        		\\
                \subfigure[\label{fig:example_1_recipe_1}]{\raisebox{0mm}{\includegraphics[width=0.95\linewidth]{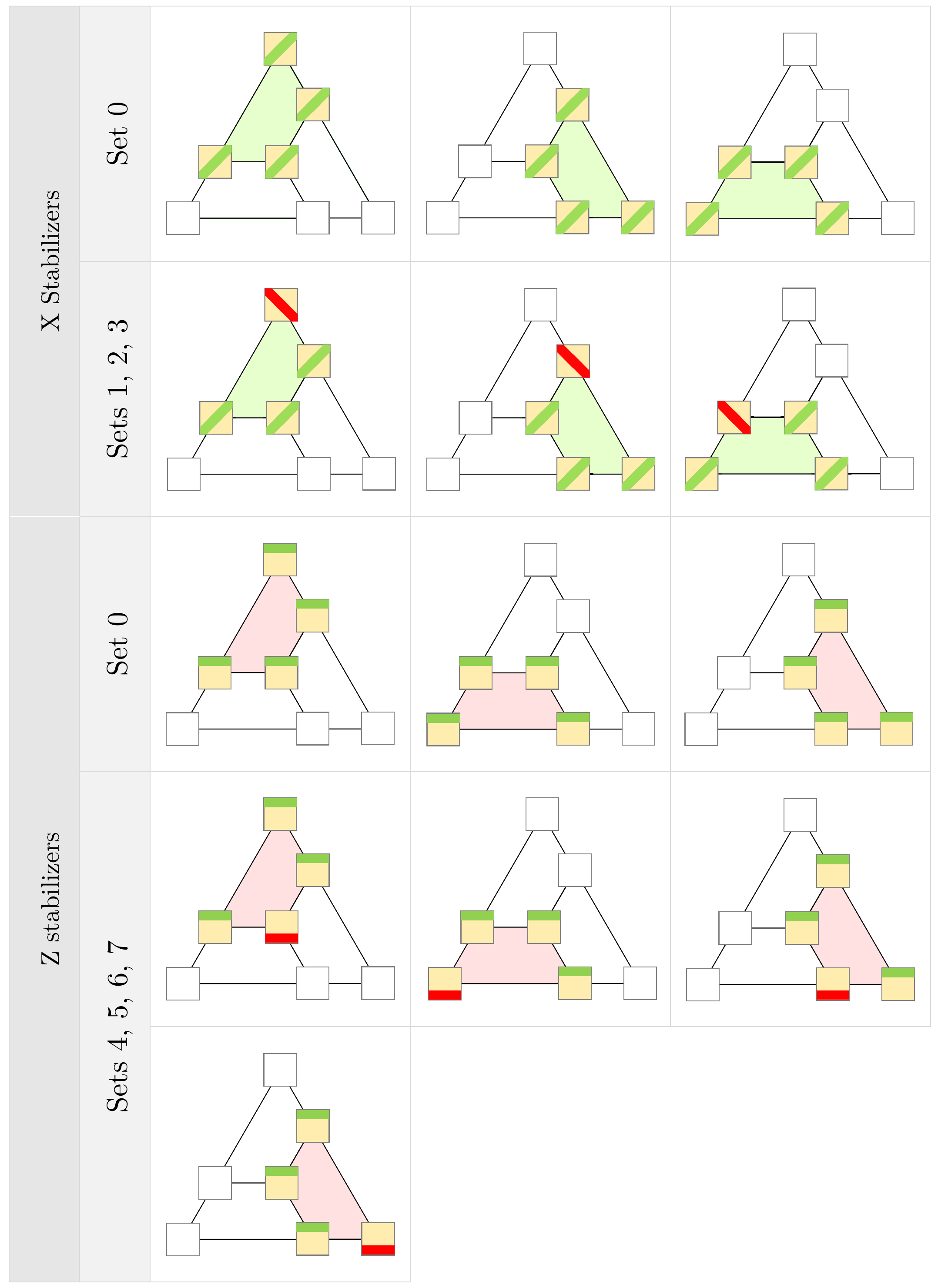}}}
                \\
                \subfigure[\label{fig:pictorial_product}]{\raisebox{0mm}{\includegraphics[width=0.65\linewidth]{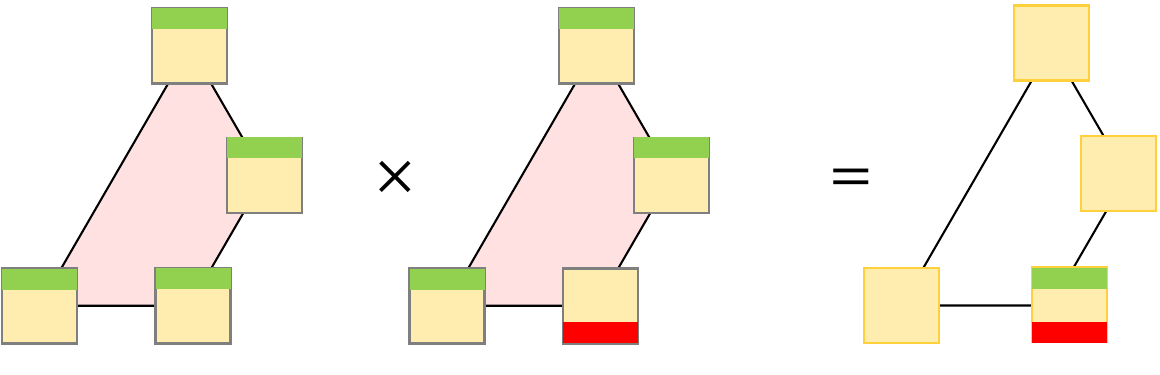}}}
                \\
        \end{tabular}
        \caption{
        	(a) The figure shows the stabilizer generators of the \(\llbracket 14,1,6_f \rrbracket\)  fermionic code. The three colored plaquettes correspond to three X stabilizers and three Z stabilizers, similar to the Steane color code. In addition, the seven tetrons also belong to the stabilizer generator group, and they are shown as yellow squares. Although the tetrons belong to the stabilizer generator group, they are not directly measurable. 
        	(b) The figure shows the measurable stabilizers of the \(\llbracket 14,1,6_f \rrbracket\)  fermionic code, derived from the \(\llbracket 7,1,3_b \rrbracket\)  bosonic color code. Each stabilizer plaquette is supported on four Majorana operators at its four vertices. These operators are defined according to \figref{fig:tetron_representation}. 
         	(c) The figure demonstrates that the stabilizer group contains one tetron operator. Similarly, all seven tetrons belong to the stabilizer group.
        } 
        \label{fig:example_1_recipe_1_and_pictorial_product}
\end{figure}

Finally, in Sec.~\ref{sec:error_analysis}, we study the code capacity of a fermionic code family, and provide an example of a fault-tolerant measurement scheme.

\section{Error correction principle}
\label{sec:principle}

We begin with a brief introduction to general Majorana operators and general Majorana error correcting codes. Then we discuss the principles of correcting errors specific to the tetron architecture. We discuss the challenges of correcting fermionic errors, and then put forward our proposal.

\subsection{Majorana operators}
\label{ssec:majorana_operators}
Consider a system with \(N_{Maj}\) MZMs \(\gamma_1, \gamma_2, \dots \gamma_{N_{Maj}}\), where \(N_{Maj}\) is even. They satisfy:
\begin{equation*}
\gamma_j = \gamma_j^\dag, \; \gamma_j^2 =\mathbb{I}, \; \{\gamma_j,\gamma_k\} = 2\delta_{j,k}
\enspace .
\end{equation*}

In this system, we can define a Majorana operator \(\gamma_S = i^{\abs{S}/2} \prod_{j \in S} \gamma_j\), which is supported on an MZM set \(S\subseteq\left\{ 1,2,\dots,N_{Maj} \right\}\). This Majorana operator may also be alternately represented as \(\ket{S}=\sum_{j \in S}\ket{j}\), where the MZM \(\ket{j}\) is a basis vector in \( \mathbb{F}_{2}^{N_{Maj}} \). The commutation relation between two Majorana operators \(\ket A\) and \(\ket B\) is given by \(\ket A \times \ket B = (-1)^{|A| \cdot |B| + |A\cap B|} \times \ket B \times \ket A \). However, parity measurement is only possible for Majorana operators which are supported on an even number of MZMs. Two even-weight Majorana operators commute only if they overlap on an even number of Majorana modes. The parity measurement of an even-weight Majorana operator \(\ket M\) yields either even parity (0) or odd parity (1). However, if it is affected by an error \(\ket E\) that anticommutes with \(\ket M\), then the parity measurement of \(\ket M\) is toggled. On the other hand, if \(\ket E\) commutes with \(\ket M\), then the parity measurement of \(\ket M\) is unaffected.

\subsection{Majorana codes}
\label{ssec:majorana_codes}
The total number of MZMs in a system must be an even integer, say \(N_{Maj}\). They can be partitioned into \(N_{Maj}/2\) MZM pairs, wherein each partitioning choice corresponds to a basis. Each MZM pair exists in a superposition of two degenerate parity eigenstates. Thus, the Majorana system has \(n=N_{Maj}/2\) degrees of freedom and \(2^n\) degenerate parity eigenstates.

Among these \(n\) degrees of freedom, we restrict \(n_s=n-k\) degrees of freedom by requiring that the allowed wavefunctions are eigenstates of \(n_s\) stabilizer operators. This leaves us with \(k\) degrees of freedom, wherein \(k\) logical qubits can be encoded. Thus, we can form an \(\llbracket n,k,d_f \rrbracket\) stabilizer code, where \(d_f\) denotes the minimum Majorana weight of all logical operators. 

\subsection{Tetrons}
\label{ssec:tetrons}

A tetron is a superconducting island with four Majorana zero modes (MZMs), maintained in an overall even parity state by high charging energy \cite{Karzig2017}. The MZMs of a tetron are denoted as \(\gamma_m\), where the character \(m\) denotes the MZM location according to \figref{fig:tetron}.

A system with \(4n\) Majorana modes would have \(2n\) degrees of freedom. But if they are implemented on a tetron architecture, then \(n\) degrees are restricted by the requirement of even parity on all \(n\) tetrons. Thus, each tetron is left with a single degree of freedom, wherein one qubit information can be stored \cite{Kitaev2006}. Each Pauli operator of this qubit has two possible representations, as shown in \figref{fig:tetron_representation}. For example, the Pauli operator \(X\) can be represented as either \(\gamma_b \gamma_c\) or \(\gamma_a \gamma_d\). Both of them yield the same parity measurement since the tetron has zero parity in total.

While addressing Majorana codes described on multiple tetrons, we would often use a numerical subscript \(q\) to denote that Majorana operators such as \(\gamma_{aq}\) or \(X_q\) correspond to the \(q\text{th}\) tetron.

\subsection{Bosonic and fermionic errors in tetrons}
\label{ssec:errors_in_tetrons}
A tetron might be affected by an error that affects 1, 2, 3 or all 4 of its MZMs. An odd-weight error is said to be a fermionic error while an even-weight error is said to be a bosonic error \cite{Viyuela2019}. Among these errors, the weight-4 error is trivial since it does not affect the parity of any Pauli operator. Also, a weight-3 error is equivalent to a weight-1 error, for example \(\gamma_a \equiv \gamma_b \gamma_c \gamma_d\). So, error correction is only required for weight-2 Pauli errors and weight-1 fermionic errors.

\subsection{Bosonic error correction in tetron architecture}
\label{ssec:bosonic_error_correction}
A tetron is maintained in an overall even parity state by a high charging energy, which suppresses fermionic errors \cite{Karzig2017,Knapp2018}. As long as the charging energy preserves even parity on tetrons, only Pauli errors can occur on a tetron, and a conventional bosonic error correcting code will be sufficient to correct it.

In a system with \(n\) tetrons maintained in even parity, we have \(n\) degrees of freedom. We can define \(n_s\) stabilizer operators which restrict \(n_s\) degrees of freedom and leaves \(k=n-n_s\) degrees for encoding \(k\) logical qubits. This forms an \(\llbracket n,k,d_b \rrbracket\) bosonic error correcting code, where \(d_b\) denotes the minimum qubit weight of all logical operators.

Note that each measurable Majorana operator should span zero or two MZMs per tetron, so as to preserve the even parity subspace of each tetron qubit. Parity measurement of all four MZMs of a tetron are not considered because of experimental challenges.

\subsection{Challenges of fermionic error correction}
\label{ssec:challenges_fermionic_error_correction}

If the charging energy is insufficient to prevent long-lived fermionic errors on tetrons, then we require a fermionic error correcting code. In a system with \(n\) tetrons, where the tetrons are susceptible to fermionic errors, there are \(2n\) degrees of freedom. Thus, a fermionic error correcting code on \(n\) tetrons would have \(\llbracket 2n,k,d_f \rrbracket\) parameters, where \(d_f\) denotes the minimum Majorana weight of all logical operators.

One possible approach is to restrict \(n\) degrees of freedom by introducing \(n\) stabilizers, where the \(q\text{th}\) stabilizer is supported on the 4 MZMs of the \(q\text{th}\) tetron, such as in Ref.~\onlinecite{Litinski2018}. The syndrome of these \(n\) stabilizers would allow us to detect fermionic errors on any tetron. When coupled with syndromes of other stabilizers, it would enable us to correct both fermionic and bosonic errors affecting the Majorana code. Unfortunately, as we have previously discussed, parity measurement of all four MZMs of a tetron has turned out to be experimentally challenging \cite{Knapp2018}.
A second approach involves dynamically varying the number of MZMs on an island, but that has significant experimental challenges as well \cite{Knapp2018}.
Ref.~\onlinecite{Bomantara2020} proposed a fermionic error correction scheme for a quantum wire, where system parameters are tuned to generate additional Majorana modes at the wire endpoints, such that stabilizers defined on these Majorana modes can correct fermionic errors. However, this can also be experimentally challenging. 

\subsection{Proposed fermionic error correction principle}
\label{ssec:fermionic_error_correction_principle}
We describe the basic principle of fermionic error correction on a single tetron. We define two complementary sets of weight-2 operators, shown in \figref{fig:tetron_representation}.
\begin{align*}
\mathcal{R}=\{X,Y,Z\}, \hspace{2.2 pt} & \text{ where } X=\gamma_b \gamma_c, Y=\gamma_a \gamma_c, Z=\gamma_a \gamma_b\\
\mathcal{R}'\hspace{-\prml}=\{X'\hspace{-\prml},Y'\hspace{-\prml},Z'\}, \hspace{-\prml} & \text{ where } X'\hspace{-\prml}=\gamma_a \gamma_d, Y'\hspace{-\prml}=\gamma_d \gamma_b, Z'\hspace{-\prml}=\gamma_c \gamma_d
\end{align*}

If a weight-2 error affects the tetron, then both sets of operators are similarly affected. For example, a weight-2 error \(\gamma_b \gamma_c\) toggles the parity of operators \(Y,Y',Z,Z'\), while leaving the \(X,X'\) operators unaffected. This forms the building block of all bosonic error correction schemes. 
However, if a weight-1 error affects the tetron, then the two sets of operators are oppositely affected. For example, a fermionic error \(\gamma_a\) toggles the parity of \(Y\) and \(Z\), while leaving \(X\) unaffected. Conversely, the same \(\gamma_a\) error leaves the parity of \(Y'\) and \(Z'\) unaffected, while toggling the parity of \(X'\). Thus, we can identify and correct a fermionic error on a tetron. This is the building block of our fermionic error correction proposal. 

\section{Review of bosonic codes}
\label{sec:bosonic_codes}
Before moving on to fermionic error correcting codes, we shall discuss how bosonic codes may be used to deal with long-lived fermionic errors, as well as its shortcomings. Consider an \(\llbracket n,k,d_b \rrbracket\)  bosonic code implemented on a tetron architecture, where each Pauli operator maps to a Majorana operator in \(\mathcal{R}\) \cite{Bravyi2010}. This code is capable of correcting up to \(t_b=\lfloor(d_b-1)/2\rfloor\) number of bosonic errors if no fermionic errors occur. Now suppose that the Majorana code is not only affected by bosonic errors on the tetron set \(T_b\), but also by fermionic errors of types \(\gamma_a,\gamma_b,\gamma_c,\gamma_d\) in the tetron sets \(T_{fa},T_{fb},T_{fc},T_{fd}\) respectively. Observe that a fermionic \(\gamma_x\) error on a tetron yields the same syndromes as a bosonic error \(\gamma_x \gamma_d\) (or its equivalent) on that tetron. Thus, fermionic errors of types \(\gamma_a,\gamma_b,\gamma_c\) are identified as \(X, Y, Z\) errors respectively, while \(\gamma_d\) errors are invisible. Thus, a \(\gamma_x\) error corresponds to a \(\gamma_x \gamma_d\) correction. Such a correction can be uniquely performed only if \(\vert T_b \vert +\vert T_{fa} \vert +\vert T_{fb} \vert +\vert T_{fc} \vert \leq t_b\).
If we correct a \(\gamma_x\) error by a \(\gamma_x \gamma_d\) correction, then we essentially shift the error from \(\gamma_x\) to \(\gamma_d\). If a fermionic error occurs at \(\gamma_d\), then it would stay there undetected. Thus, the \(d\text{th}\) MZM would act as a reservoir for fermionic errors at all locations of the tetron.

It might appear that the remnant \(\gamma_d\) errors do not affect either stabilizers or logical operators, and hence they can be left uncorrected. However, Ref.~\onlinecite{Knapp2018} notes that such uncorrected excitations may propagate from one tetron to other neighboring tetrons via connecting measurements, and finally culminate in errors of high weight. To efficiently mitigate the spread of such excitations, it recommends the usage of fermionic codes.

\section{B \texorpdfstring{\(\boldsymbol{\mapsto}\)}{to} F codes: Fermionic codes from bosonic codes}
\label{sec:fermionic_codes_1}

We show that a bosonic code with parameters \(\llbracket n,k,d_b \rrbracket\)  can be translated into a fermionic Majorana code with parameters \(\llbracket 2n,k,d_f \rrbracket\)  where \(d_f=2d_b\).

\subsection{Recipe for fermionic code construction}
\label{ssec:recipe_1}
We choose a bosonic stabilizer code \(\llbracket n,k,d_b \rrbracket\)  from any scalable code family. We can derive a new Majorana fermion code from the \(n_s\) stabilizers of the bosonic code. The Majorana fermion code contains several stabilizers, which we group into overlapping sets for convenience. 
\begin{itemize}
\item 
Set 0 contains all the \(n_s\) stabilizers of the bosonic code, wherein every Pauli operator of the bosonic code maps to Majorana operators in \(\mathcal R\).
\item
Set \(i\) contains all the \(n_s\) stabilizers of the bosonic code, wherein Pauli operators of the \(i\text{th}\) qubit map to Majorana operators in \(\mathcal R'\), and all other Pauli operators map to Majorana operators in \(\mathcal R\). Here, \(i\) varies from 1 to \(n\).
\end{itemize}
Thus, we have \(n+1\) overlapping stabilizer sets, which form a stabilizer group of rank \(2n-k\). This stabilizer group characterizes the new Majorana fermion code with \(\llbracket 2n,k,d_f \rrbracket\)  parameters.

We provide an example of this scheme in \figref{fig:example_1_recipe_1}, which shows the \(\llbracket 14,1,6_f \rrbracket\)  fermionic code derived from a \(\llbracket 7,1,3_b \rrbracket\)  bosonic code. Its logical qubit is characterized by \(\bar X = X_1 X_3 X_5, \bar Y = Y_1 Y_2 Y_7, \bar Z = Z_5 Z_6 Z_7\).

\begin{claim}
 Any combination of fermionic error yields a non-zero syndrome on a \(B \mapsto F\) Majorana fermion code.
\end{claim}

\begin{proof}
Suppose a Majorana fermion code is affected by some fermionic errors and optionally some more bosonic errors. The fermionic error would anticommute with one or more tetron operators, as they have odd intersection. The tetron operator belongs to the stabilizer group, as the \(i{\text{th}}\) tetron operator can be obtained by multiplying the independent stabilizer in set \(i\) with the corresponding stabilizer in set 0, as shown in \figref{fig:pictorial_product}. Thus, any fermionic error would yield non-zero syndrome on a \(B \mapsto F\) code. If the \(i\text{th}\) tetron is affected by fermionic error, the corresponding operators in \(\mathcal R\) and \( \mathcal R'\) would yield opposite parities, and so we would measure opposite parities on the independent stabilizer of set \(i\) and the corresponding set 0 stabilizer.
\end{proof}

\subsection{Code distance}
\label{ssec:recipe_1_distance}
In this section, we show that the derived fermionic codes have the same logical operator as the bosonic code. Hence, the fermionic code distance, or the least Majorana weight of its logical operators, is given by twice the Pauli distance of the bosonic code.  

\begin{claim}
If the bosonic code has \(\llbracket n,k,d_b \rrbracket\)  parameters, then the resultant fermionic code has \(\llbracket 2n,k,d_f \rrbracket\)  parameters, where \(d_f=2d_b\).
\end{claim}
\begin{proof}
The bosonic code has \(n\) degrees of freedom and \(n-k\) stabilizer generators, so it has \(k\) logical operators. 
The fermionic code has \(2n\) degrees of freedom, and \(2n-k\) stabilizer generators. The fermionic code has \(n\) additional tetron stabilizers, as compared to the bosonic code. Thus, the fermionic code also has \(k\) logical operators. 

Next, we note that the \(k\) logical operators in the bosonic code are also valid logical operators for the fermionic code. This is true because the set 0 stabilizers of the fermionic code are the same as the stabilizers of the bosonic code. Furthermore, each logical operator of the bosonic code can be mapped to a Majorana operator that spans 2 MZMs per tetron, and hence it commutes with each tetron stabilizer. So, the \(k\) logical operators of the bosonic code satisfy all stabilizers of the fermionic code, and are the same as the \(k\) logical operators of the fermionic code. 

Finally, the code distance of the bosonic code is \(d_b\), meaning that the logical operators have a least weight of \(d_b\) Pauli operators. Since the same set of logical operators are shared between both codes, and each Pauli operator can be mapped to 2 MZMs, so the fermionic code has code distance \(d_{f} = 2d_{b}\), where \(d_f\) is the least Majorana weight of its logical operators. 
\end{proof}

\subsection{Decoder}
\label{ssec:recipe_1_decoder}
We use the BPOSD decoder, proposed by Roffe et al.~\cite{Roffe2020, Roffe2022} for error correction. A code with distance \(d_f\) can correct any error that has Majorana weight \(< d_{f}/2\). As the code size increases, so does the code distance \(d_f\) and the error correction capacity.

\subsection{Error correction on color codes}
\label{ssec:recipe_1_latency_color}
In this section, we shall discuss Majorana fermion codes derived from CSS color codes. The error correction latency of a tetron code is limited by the fact that only operators with disjoint tetron support can be parallelly measured. Although it is theoretically possible to measure two copies of the same Pauli operator on a tetron, it has been recommended that such measurements be avoided to prevent the spread of correlated errors \cite{Knapp2018}. The CSS color code based on the 6.6.6 tessellation  is 3-colorable, so the corresponding set 0 stabilizers require 3 steps for measuring all X stabilizers and 3 steps for measuring all Z stabilizers.

If the Majorana fermion code is described on \(n\) tetrons, then it would have \(n\) independent stabilizers in addition to the set 0 stabilizers. In each of these \(n\) additional stabilizers, one of the tetron operators is switched from \(\mathcal R\) to \(\mathcal R'\). These additional stabilizer syndromes can be extracted in seven steps. So, the fermionic error correction latency is 13.
\begin{itemize}
\item 
The \(\llbracket 14,1,6_f \rrbracket\)  code has 13 independent stabilizers, each of which overlaps with the other, as shown in \figref{fig:example_1_recipe_1_and_pictorial_product}. Thus, one round of syndrome extraction requires 13 steps, and the error correction latency is 13. 
\item
The \(\llbracket 38,1,10_f \rrbracket\)  code has 37 independent stabilizers. Set 0 contains 9 \(X\) stabilizers, 9 \(Z\) stabilizers, and it needs 3 + 3 steps to be measured. Sets 1 to 19 contain one independent stabilizer each, and they need 7 steps to be parallelly measured, as shown in \figref{fig:example_2_recipe_1}. Thus, the error correction latency is 13. 
\item
The \(\llbracket 74,1,14_f \rrbracket\)  code has 73 independent stabilizers. Set 0 contains 18 \(X\) stabilizers and 18 \(Z\) stabilizers, and it needs 3 + 3 steps to be measured. Sets 1 to 37 contain one independent stabilizer each, and they need 7 steps to be parallelly measured, as shown in \figref{fig:example_3_recipe_1}. Thus, the error correction latency is 13. 
\end{itemize}

\begin{figure}[htb]
	{\includegraphics[width=0.82\linewidth]{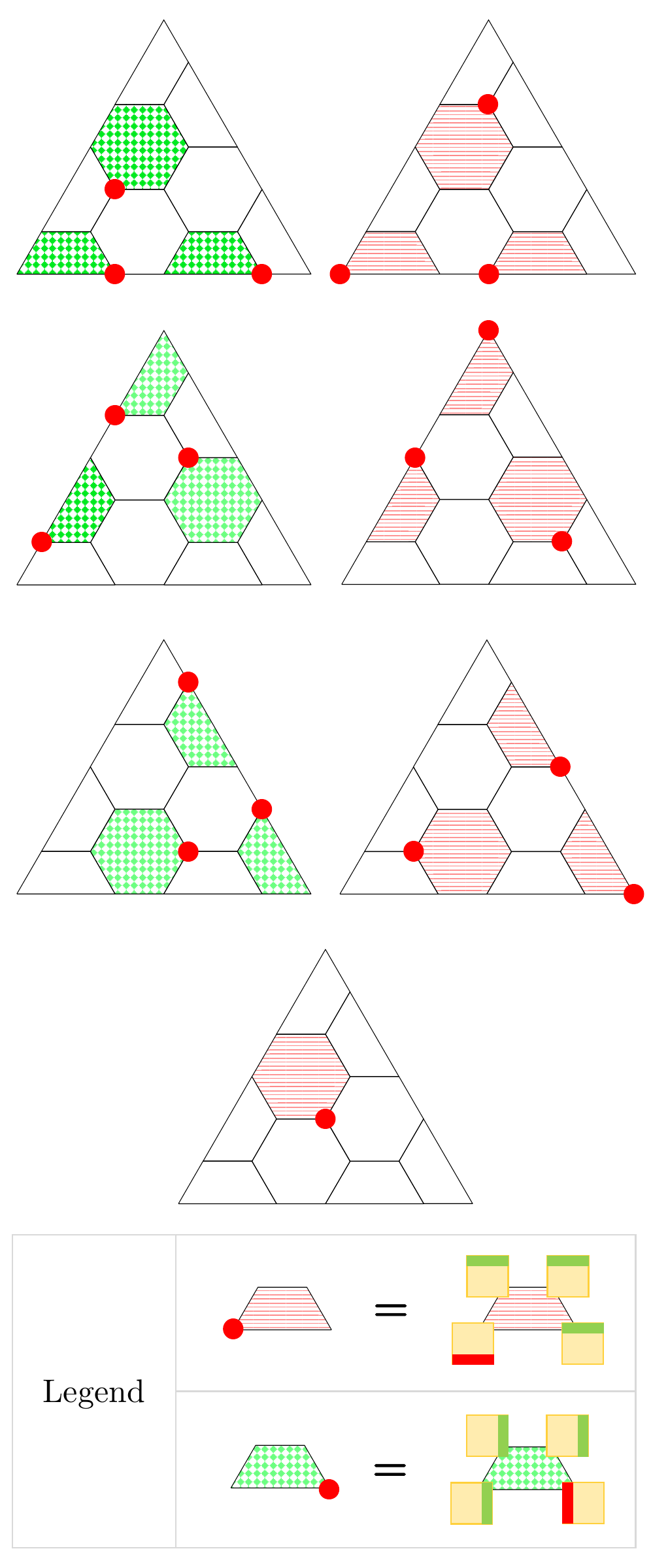}}
	\caption{
		The figure illustrates an optimized syndrome measurement sequence for stabilizers in sets 1 to 19 of \(\llbracket 38,1,10_f \rrbracket\) code. Each pink striped plaquette corresponds to a stabilizer supported on \(Z'\) at the vertex marked by a red circle, and supported on \(Z\) at all other vertices of that plaquette. Similarly, each green checkered plaquette corresponds to a stabilizer supported on \(X'\) at the vertex marked by a red circle, and supported on \(X\) at all other vertices of that plaquette.} 
	\label{fig:example_2_recipe_1}
\end{figure}

\begin{figure}[htb]
	{\includegraphics[width=\linewidth]{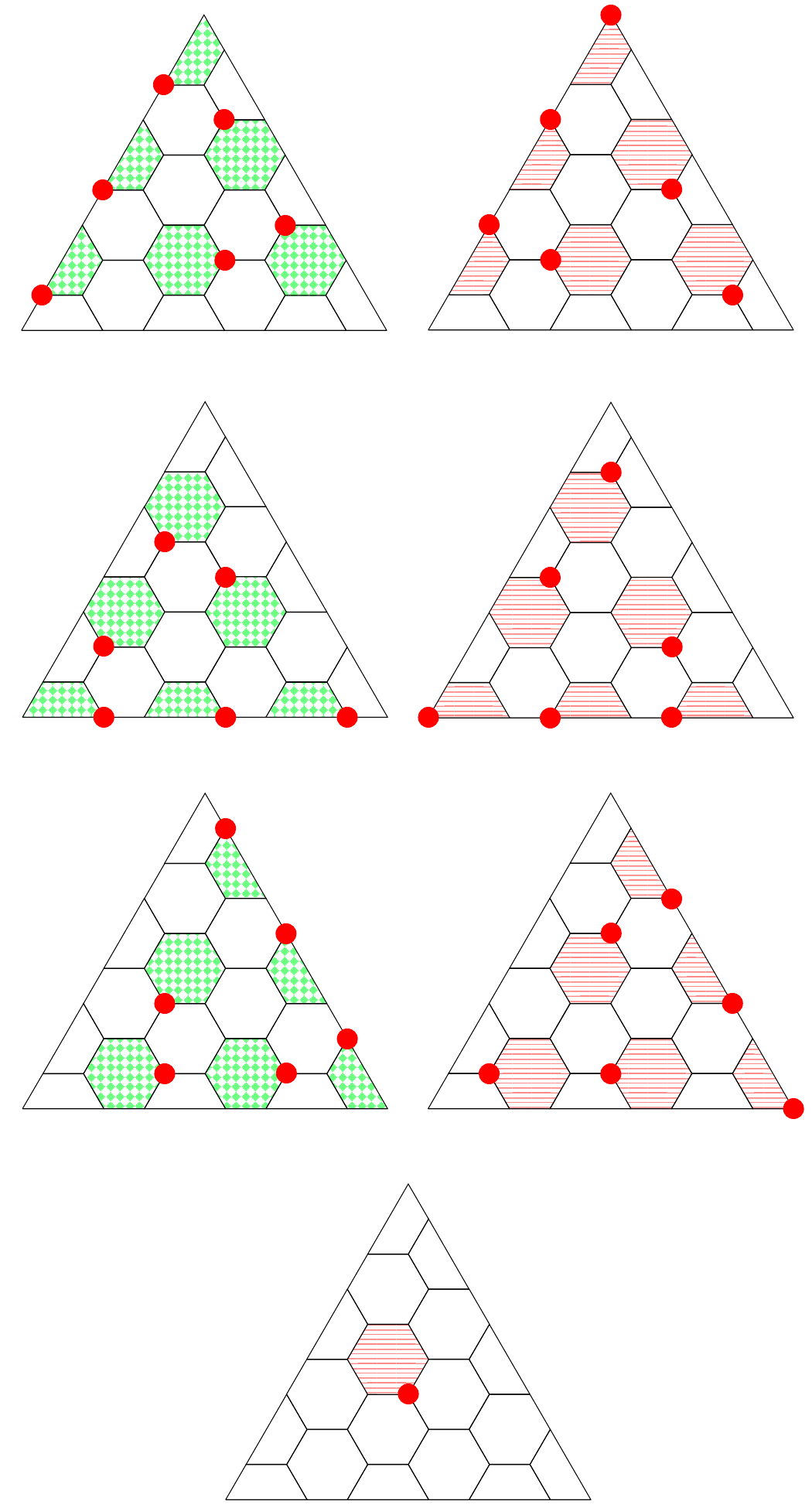}}
	\caption{
		The figure illustrates an optimized syndrome measurement sequence for stabilizers in sets 1 to 37 of \(\llbracket 74,1,14_f \rrbracket\) code. Each pink striped plaquette corresponds to a stabilizer supported on \(Z'\) at the vertex marked by a red circle, and supported on \(Z\) at all other vertices of that plaquette. Similarly, each green checkered plaquette corresponds to a stabilizer supported on \(X'\) at the vertex marked by a red circle, and supported on \(X\) at all other vertices of that plaquette.}
	\label{fig:example_3_recipe_1}
\end{figure}

\begin{figure}[t]
	\begin{tabular}{c}
		\subfigure[\label{fig:rotated_surface_code}]{\includegraphics[width=0.6\linewidth]{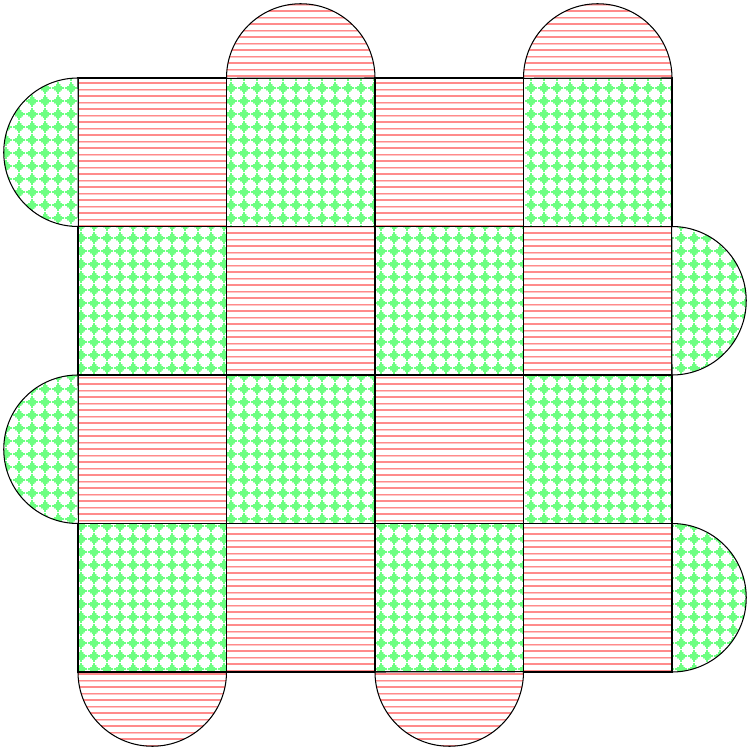}}
		\\
		\subfigure[\label{fig:rotated_surface_code_fermionic_syndrome}]{\includegraphics[width=0.9\linewidth]{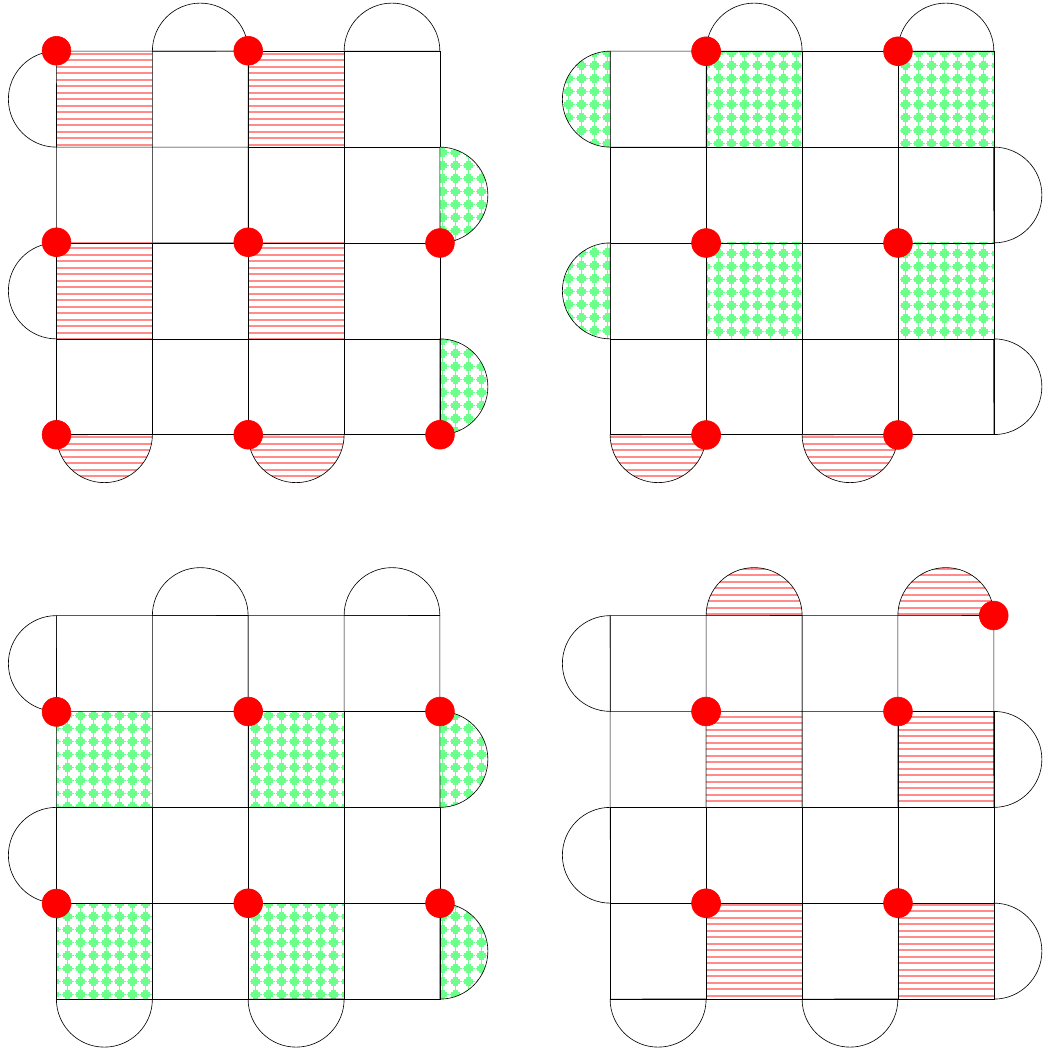}}
	\end{tabular}
	\caption{
		(a) The figure shows the stabilizers of the \(\llbracket 25,1,5_b \rrbracket\) code, which are the same as the set 0 stabilizers of the \(\llbracket 50,1,10_f \rrbracket\) code. The pink striped stabilizer plaquettes are supported on \(Z\) operators at its vertices, while the green checkered stabilizer plaquettes are supported on \(X\) operators at its vertices. The \(X\) stabilizers of set 0 require 2 steps for syndrome extraction, similarly the \(Z\) stabilizers of set 0 require another 2 steps for syndrome extraction.
		(b) The figure illustrates an optimized four-step syndrome measurement sequence for stabilizers in sets 1 to 25 of \(\llbracket 50,1,10_f \rrbracket\) code. Each pink striped plaquette corresponds to a stabilizer supported on \(Z'\) at the vertex marked with a red circle, and supported on \(Z\) at all other vertices of that plaquette. Similarly, each green checkered plaquette corresponds to a stabilizer supported on \(X'\) at the vertex marked with a red circle, and supported on \(X\) at all other vertices of that plaquette.} 
	\label{fig:surface_code_and_steps}
\end{figure}

Before proceeding to code capacity and fault-tolerance analysis, let us see another example of this fermionic code construction recipe.

\subsection{Error correction on surface codes}
\label{ssec:recipe_1_latency_surface}
In this section, we consider Majorana fermion codes derived from rotated surface codes. Its latency for fermionic error correction is 8. 

If the Majorana fermion code is described on \(n\) tetrons, then it would have additional \(n\) independent stabilizers on top of the set 0 stabilizers. In each of these \(n\) additional stabilizers, one of the tetron operators is switched from \(\mathcal R\) to \(\mathcal R'\). This defines the stabilizers for the fermionic version of the rotated surface code.

For example, let us consider the \(\llbracket 25,1,5_b \rrbracket\) rotated surface code, shown in \figref{fig:rotated_surface_code}. From this, we derive the \(\llbracket 50,1,10_f \rrbracket\) code, containing 49 independent stabilizers. Set 0 contains 12 \(X\) stabilizers and 12 \(Z\) stabilizers, and it needs 2 + 2 steps to be measured. Sets 1 to 25 contain one independent stabilizer each, and they need 4 steps to be parallelly measured, as shown in \figref{fig:rotated_surface_code_fermionic_syndrome}. Thus, one round of fermionic syndrome extraction requires 8 syndrome measurement steps, and so the fermionic error correction latency is 8.

Thus, we have analyzed the \(B \mapsto F\) codes, wherein the stabilizer group is generated by toggling the tetron operators at one tetron in one measurement operator at a time. The tetrons are included in the resultant stabilizer group. 

\section{Error analysis}
\label{sec:error_analysis}

\subsection{Code capacity}
\label{ssec:code_capacity}

We consider a simple noise model where each tetron is subjected to bosonic errors as well as fermionic errors. If the total error rate is \(p\) and the noise bias is \(\eta\), then the bosonic error rate is \(p_b = \frac{p}{\eta + 1} \) and the fermionic error rate is \( p_f = \frac{p \eta}{ \eta + 1 } \). As the bosonic errors X, Y, Z are equally likely, so \(p_X=p_Y=p_Z= \frac{p}{3(\eta + 1)} \). Similarly, the fermionic errors \( \gamma_a , \gamma_b, \gamma_c, \gamma_d \) are equally likely, so \( p_a=p_b=p_c=p_d= \frac{p \eta }{4(\eta + 1)} \). The physical error rate is given by \( p_b+ \frac{3p_f}{4}\) since all bosonic errors affect a physical qubit, but only 3 out 4 fermionic errors affect a physical qubit on a tetron. The logical error rate for the fermionic color code family is evaluated by using the BPOSD decoder for bias values \(\eta = 0.1, 1, 10\), and are plotted in \figref{fig:code_capacity}. The corresponding variation of pseudothreshold with noise bias is shown in \figref{fig:threshold_vs_bias}.

\begin{figure*}[htpb]
	\begin{tabular}{c c c}
		\hspace{0 pt}
		\subfigure[\label{fig:cc_tetron_7_bias_01}]{\includegraphics[width=0.317\linewidth]{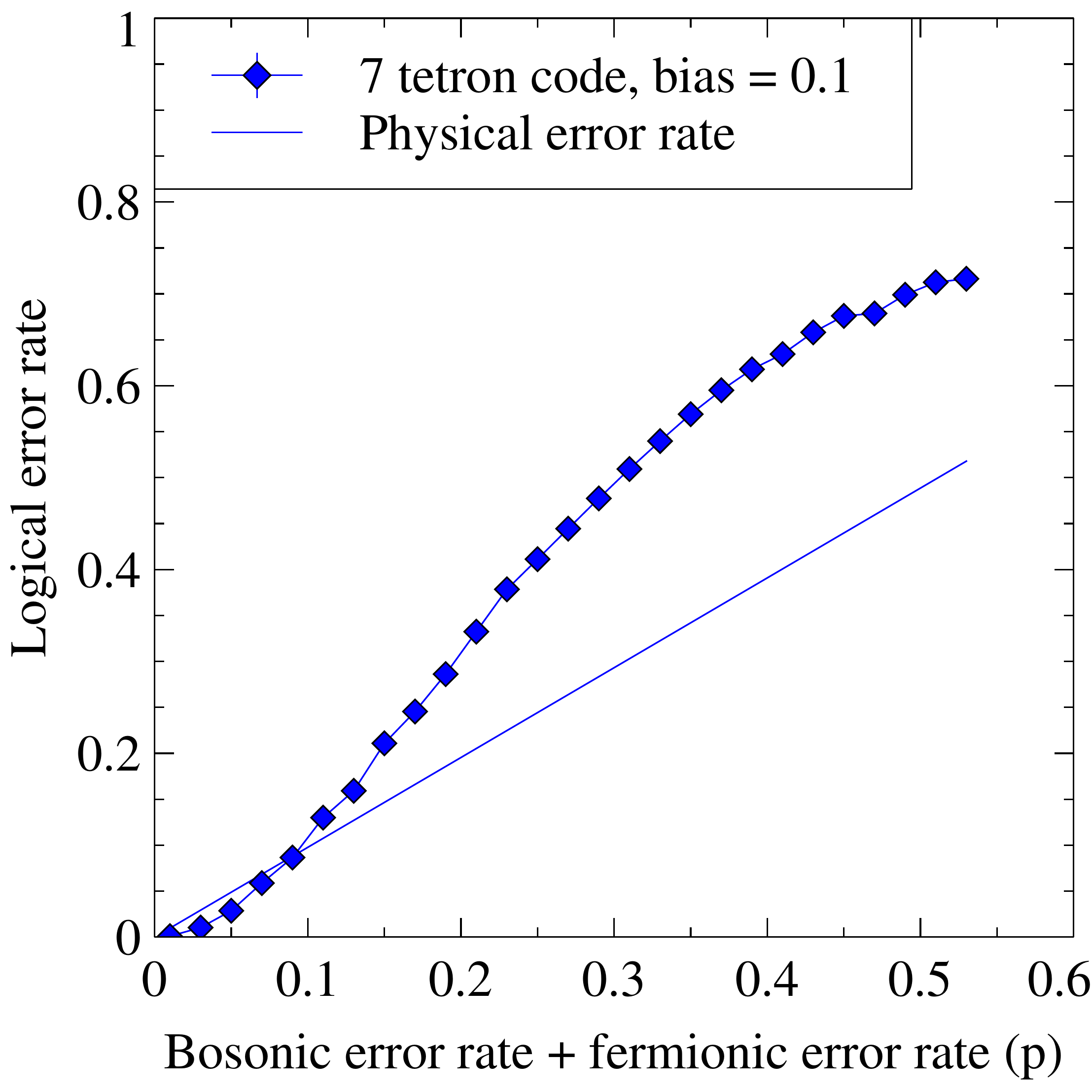}}
		\hspace{0 pt}
		&
		\subfigure[\label{fig:cc_tetron_19_bias_01}]{\includegraphics[width=0.317\linewidth]{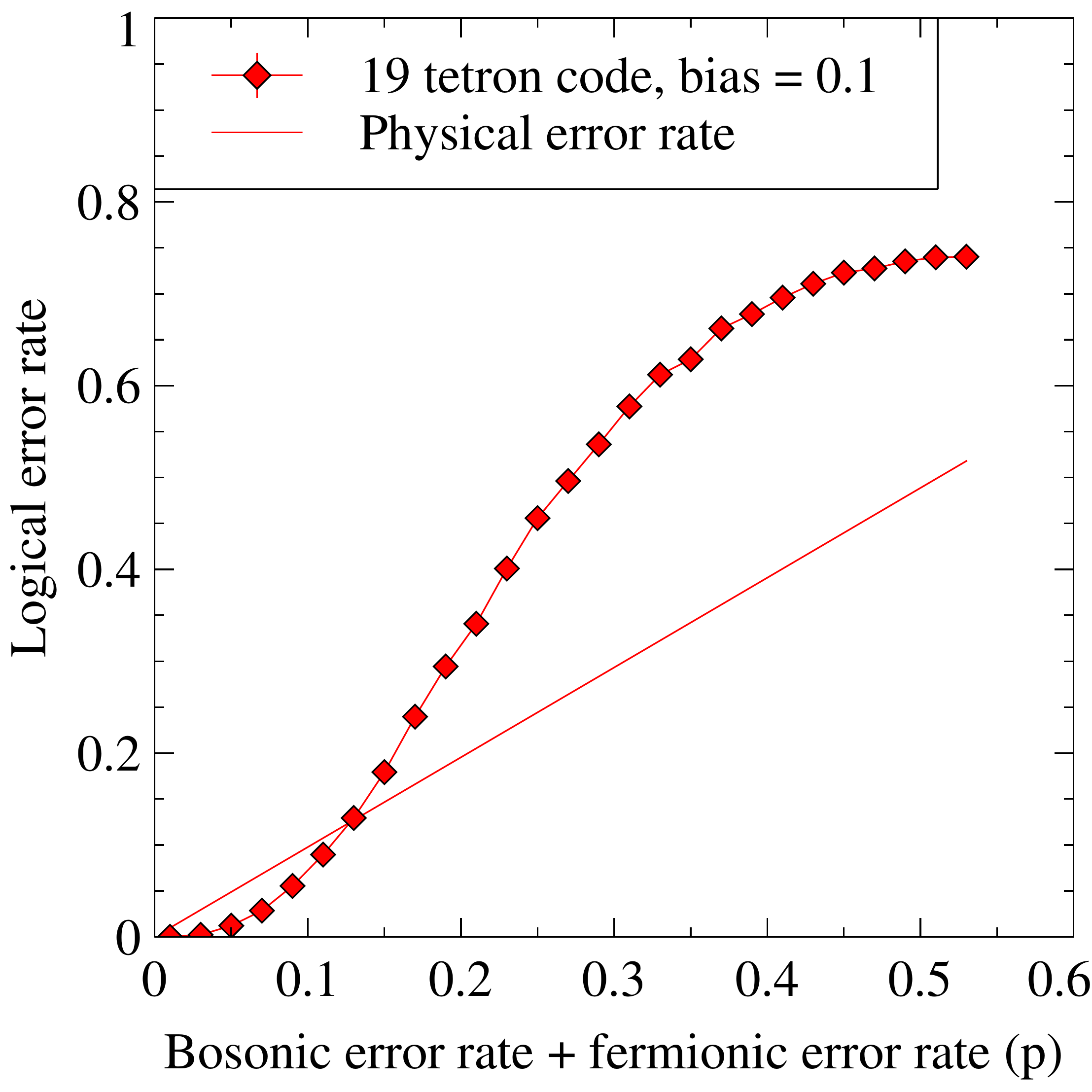}}
		\hspace{0pt}
		&
		\subfigure[\label{fig:cc_terton_37_bias_01}]{\includegraphics[width=0.317\linewidth]{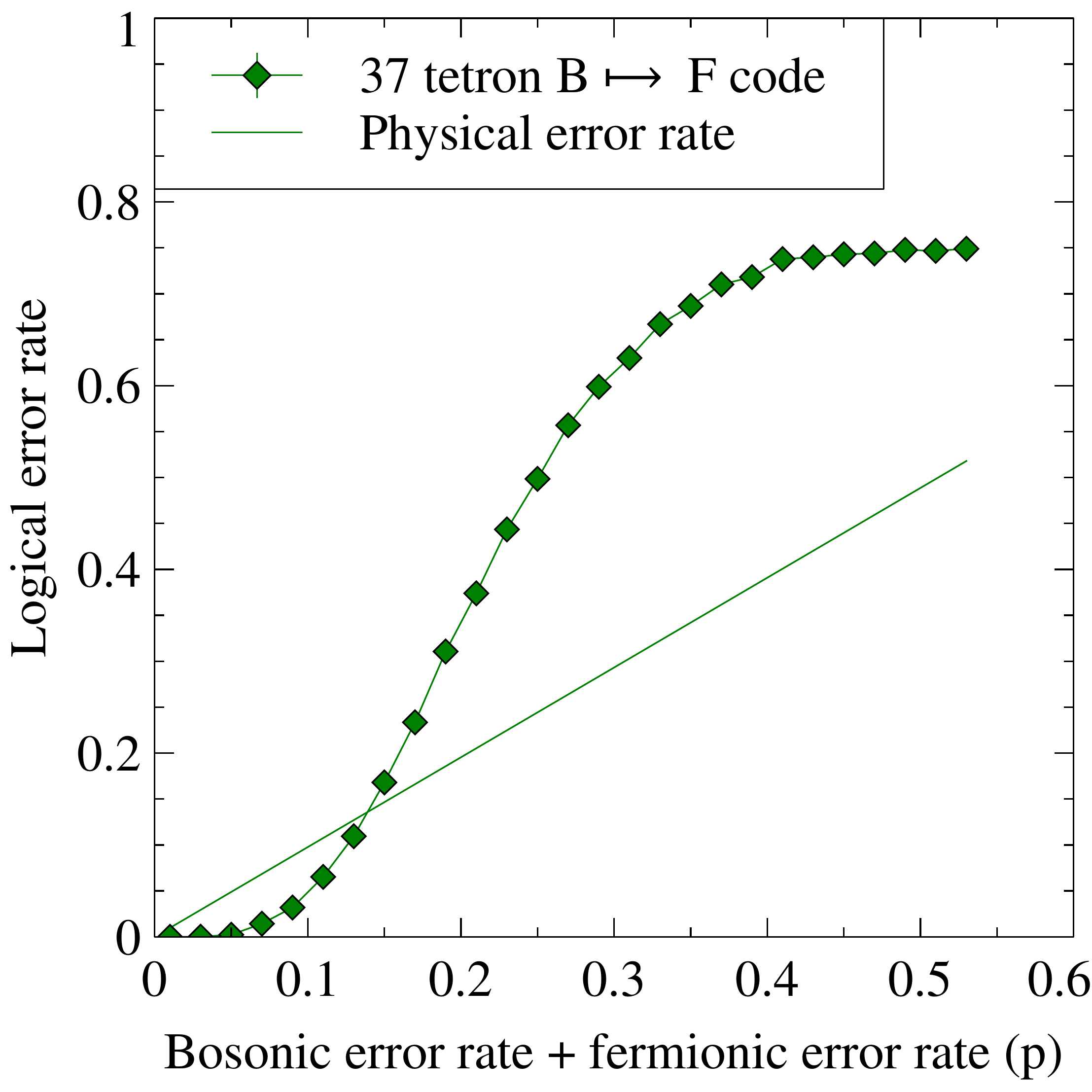}}
		\hspace{0pt}
		\\
		\hspace{0 pt}
		\subfigure[\label{fig:cc_tetron_7_bias_1}]{\includegraphics[width=0.317\linewidth]{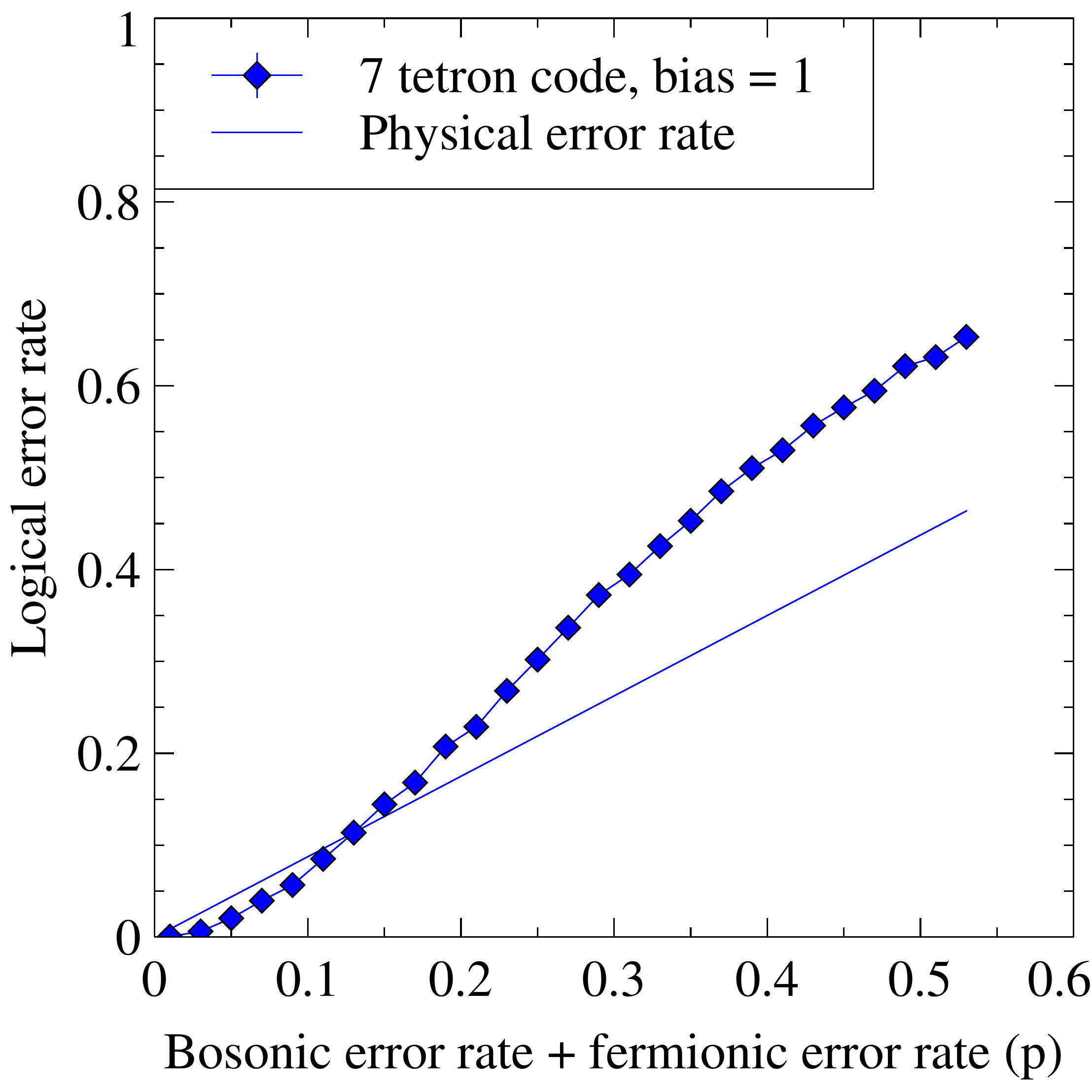}}
		\hspace{0 pt}
		&
		\subfigure[\label{fig:cc_tetron_19_bias_1}]{\includegraphics[width=0.317\linewidth]{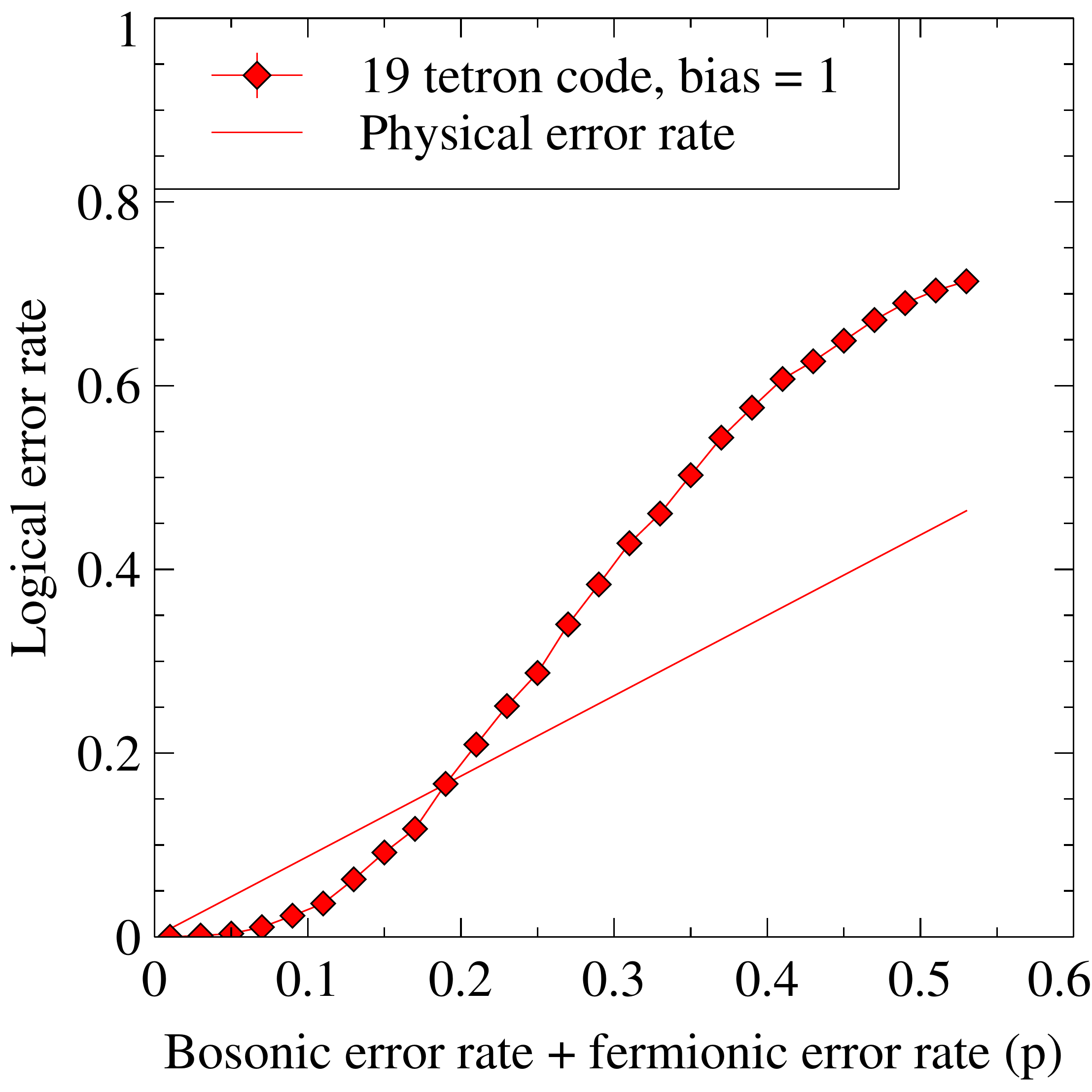}}
		\hspace{0pt}
		&
		\subfigure[\label{fig:cc_terton_37_bias_1}]{\includegraphics[width=0.317\linewidth]{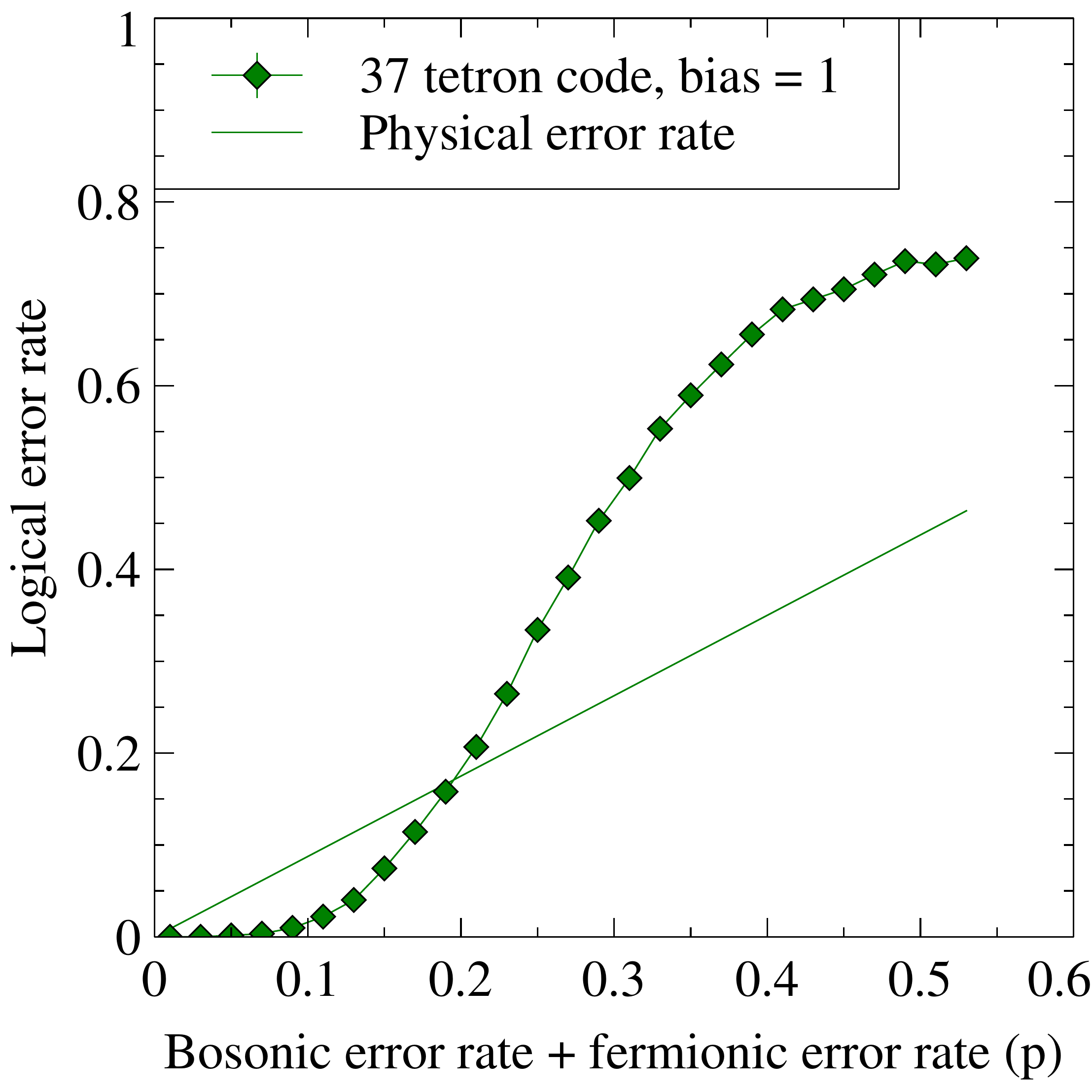}}
		\hspace{0pt}
		\\
		\hspace{0 pt}
		\subfigure[\label{fig:cc_tetron_7_bias_10}]{\includegraphics[width=0.317\linewidth]{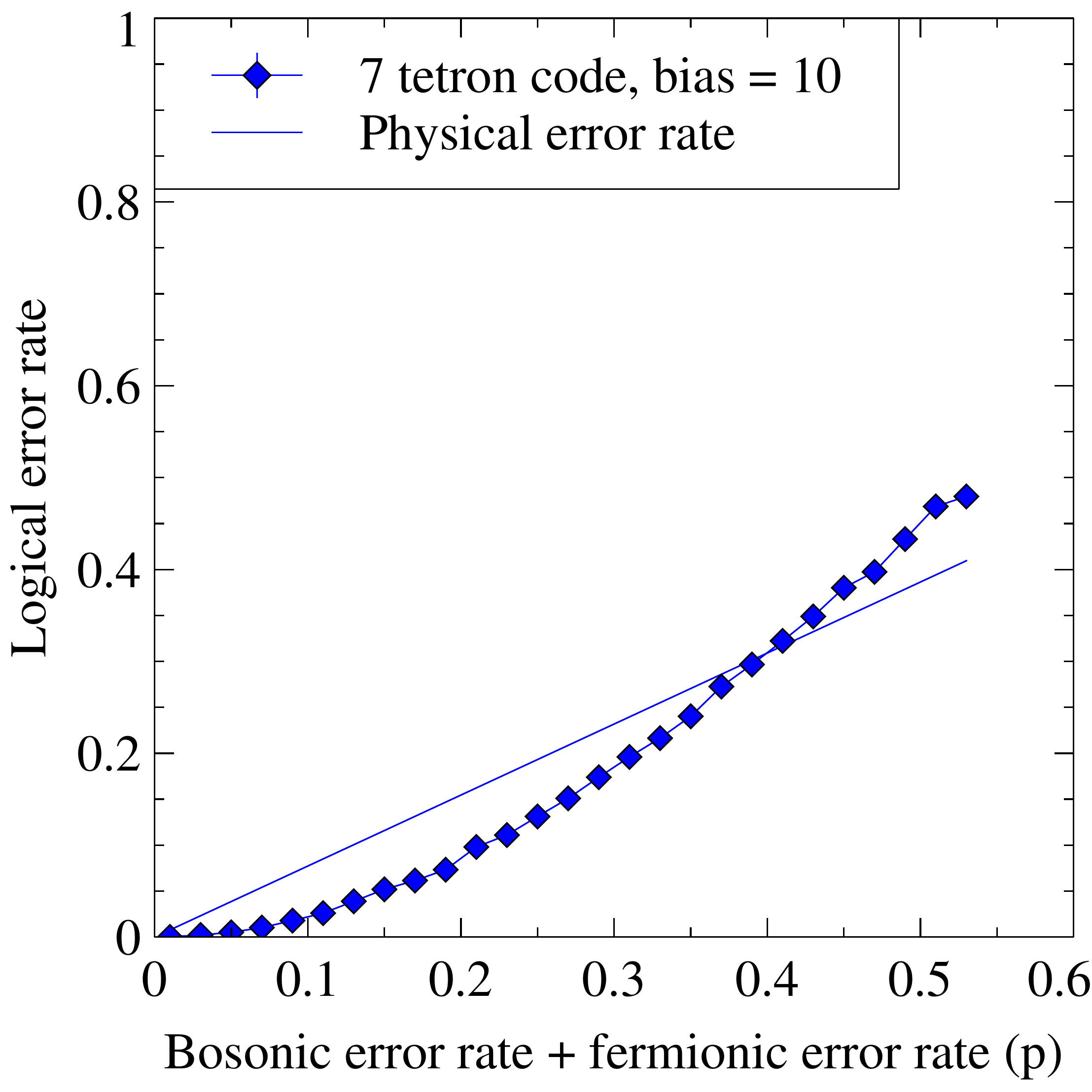}}
		\hspace{0 pt}
		&
		\subfigure[\label{fig:cc_tetron_19_bias_10}]{\includegraphics[width=0.317\linewidth]{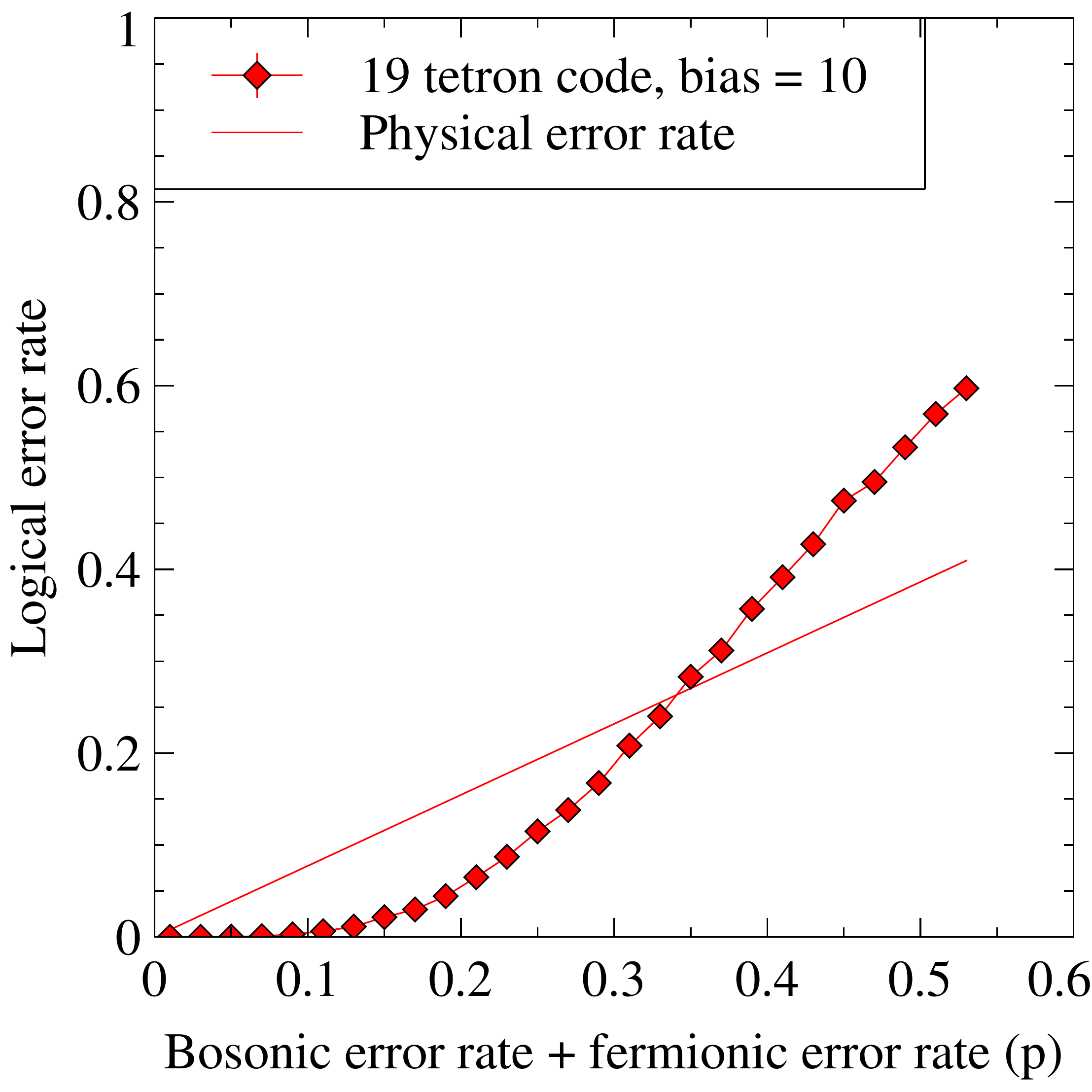}}
		\hspace{0pt}
		&
		\subfigure[\label{fig:cc_terton_37_bias_10}]{\includegraphics[width=0.317\linewidth]{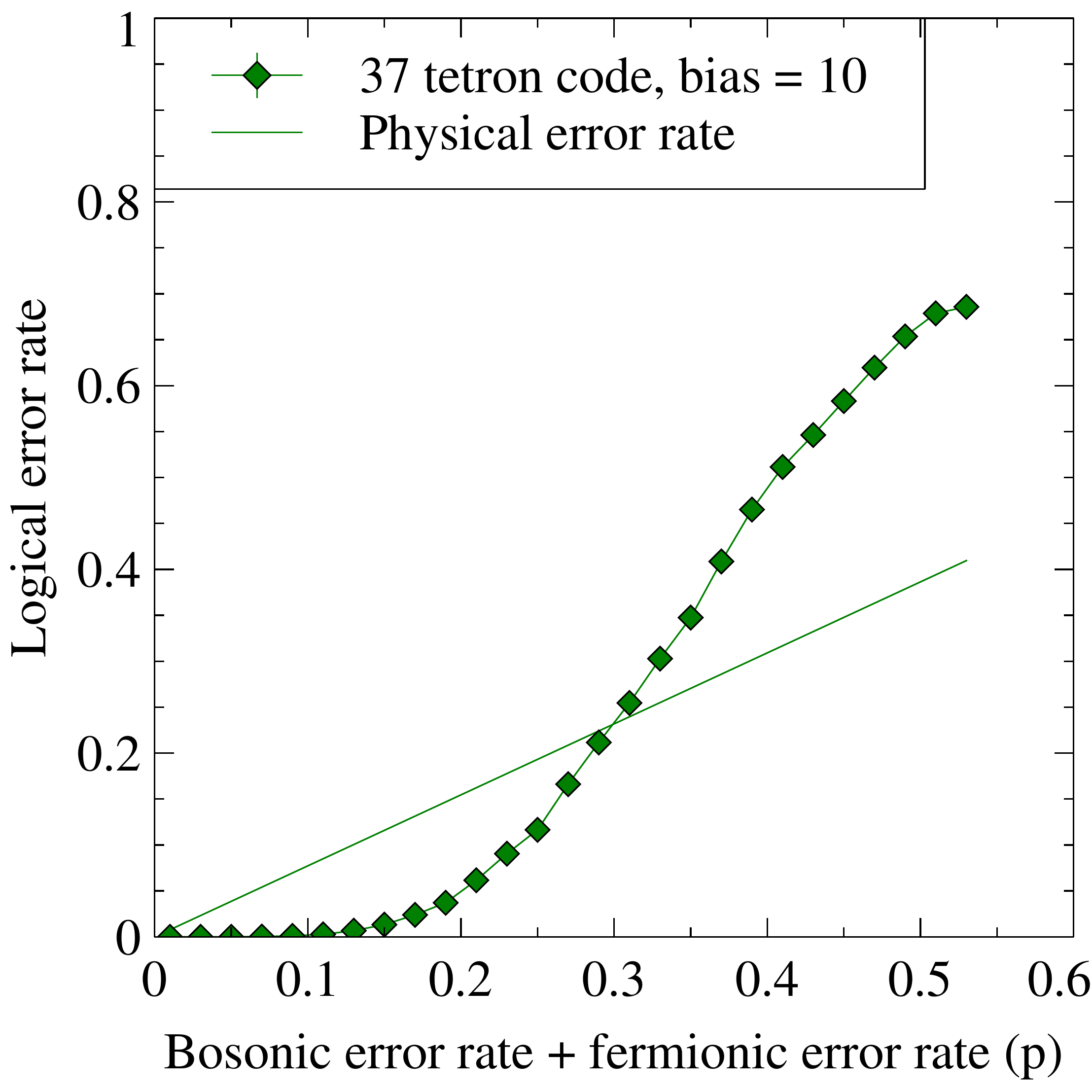}}
		\hspace{0pt}
	\end{tabular}
	\caption{
		Logical error plots for fermionic color codes with 7 tetrons (left column), 19 tetrons (middle column) and 37 tetrons (right column), along with 95\% confidence interval bars. In this case, the confidence interval bars are smaller than the marker size. 
		[Top row] The three plots show the variation of logical error rate with \(p\) for bias \( \eta = 0.1 \). 
		[Middle row] The three plots show the variation of logical error rate with \(p\) for bias \( \eta = 1 \).
		[Bottom row] The three plots show the variation of logical error rate with \(p\) for bias \( \eta = 10 \).
	} 
	\label{fig:code_capacity}
\end{figure*}

\begin{figure*}[htb]
	{\includegraphics[width=0.8\linewidth]{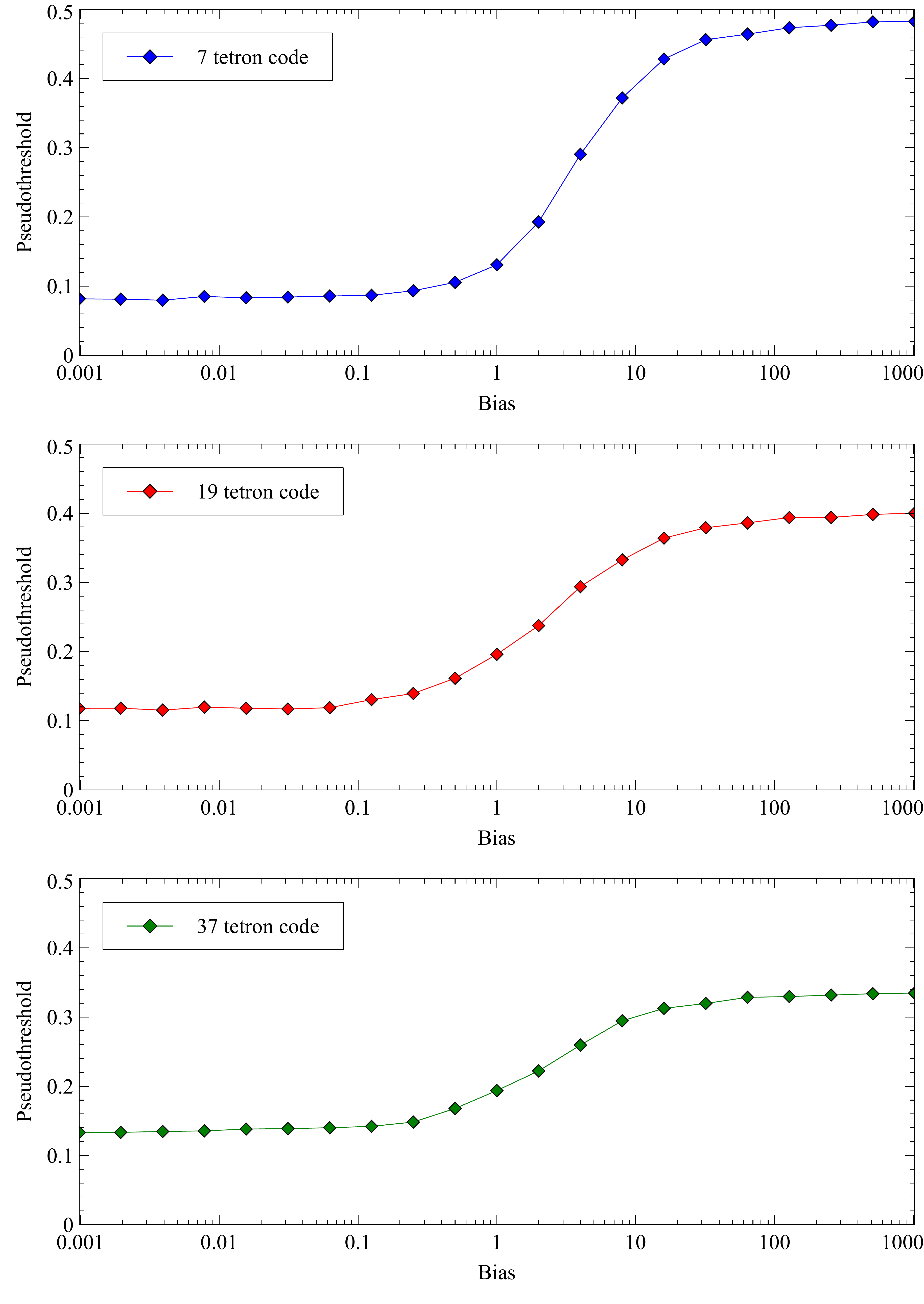}}
	\caption{
		The figure shows the variation of pseudothreshold with noise bias.
		\newline 
		The top graph shows the variation for the 7 tetron fermionic color code. 
		\newline 
		The middle graph shows the variation for the 19 tetron fermionic color code.  
		\newline
		The bottom graph shows the variation for the 37 tetron fermionic color code.  
	} 
	\label{fig:threshold_vs_bias}
\end{figure*}

\subsection{Fault tolerance}
\label{ssec:fault-tolerance}

Fault-tolerance can be achieved by additional syndrome measurements of redundant stabilizers. 

For example, \figref{fig:nsf_ft_sequence} shows a fault-tolerant sequence for the \( \llbracket 14,1,6_f \rrbracket \) fermionic color code, that can tolerate 1 input error or 1 intermediate error, either bosonic or fermionic. We analyze this sequence against a noise with \( \eta \) bias, a bosonic error rate of \(p_b = \frac{p}{\eta + 1} \), a fermionic error rate of \( p_f = \frac{p \eta}{ \eta + 1 } \), and a measurement error rate of \(p\). \Figref{fig:nsf_ft_threshold} shows the existence of a fault-tolerant threshold for noise bias \( \eta = 0.1, 1, 10\).

\begin{figure*}[htpb]
	{\includegraphics[width=0.697\linewidth]{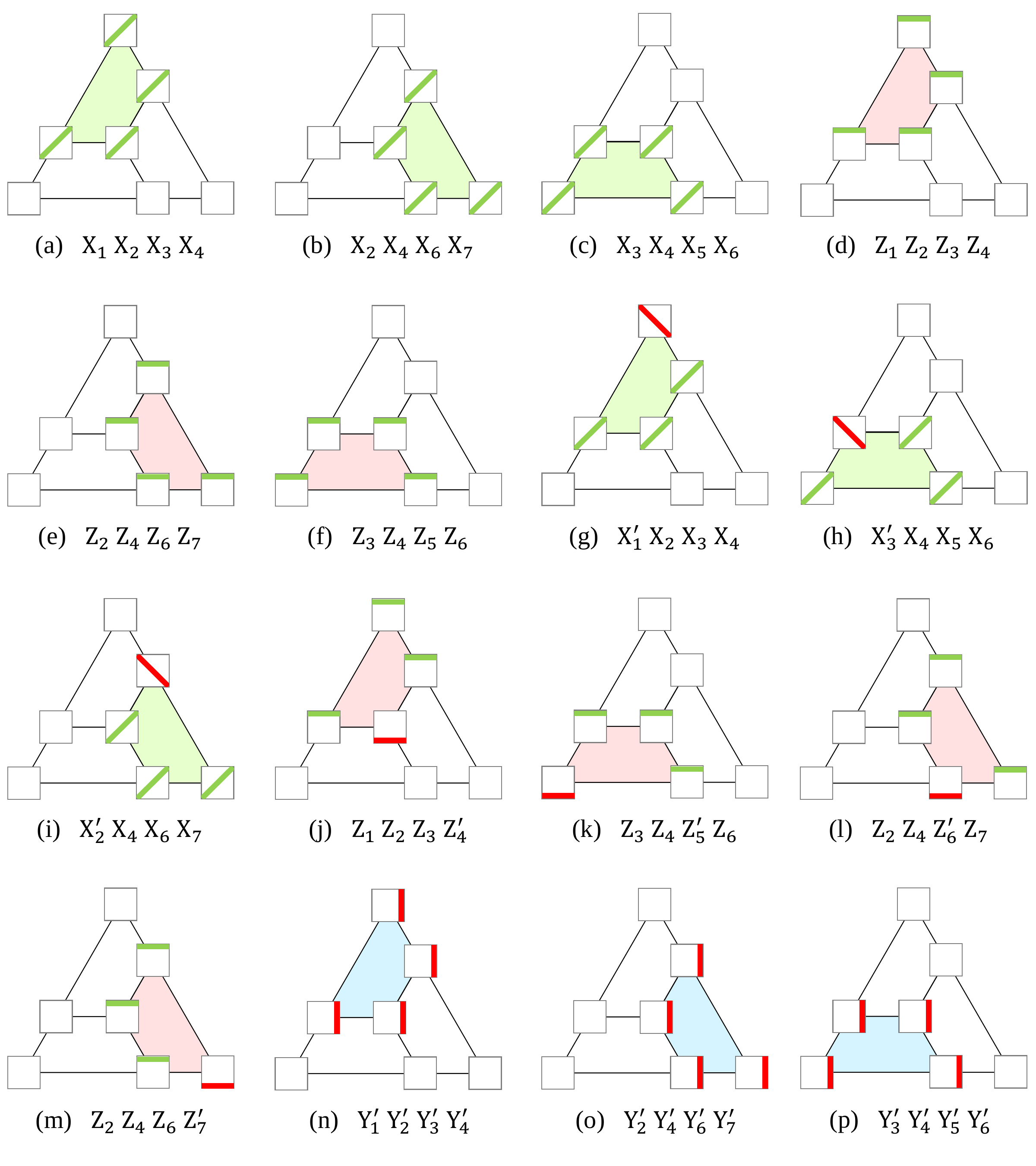}}
	\caption{
		Figures (a) to (p) demonstrate a sequence of fault-tolerant measurements for the \( \llbracket 14,1,6_f \rrbracket \) fermionic code. This sequence can tolerate one bosonic error or one fermionic error. In each step, a single stabilizer is measured, which is shaded in color. Each stabilizer is supported on four tetron operators at its vertices. A tetron is indicated by a square since it has 4 MZMs at 4 corners, and a tetron operator (such as \(Z = \gamma_a \gamma_b \)) is shown by highlighting one edge of the square (such as the top edge that connects \(\gamma_a\) and \(\gamma_b\)). The tetron operators are illustrated by the same conventions as \figref{fig:tetron_representation}. The tetrons are numbered according to \figref{fig:example_1_generators}.
	} 
	\label{fig:nsf_ft_sequence}
\end{figure*}

\begin{figure*}[htpb]
	\begin{tabular}{c c c}
		\hspace{0 pt}
		\subfigure[\label{fig:nsf_ft_bias_01}]{\includegraphics[width=0.317\linewidth]{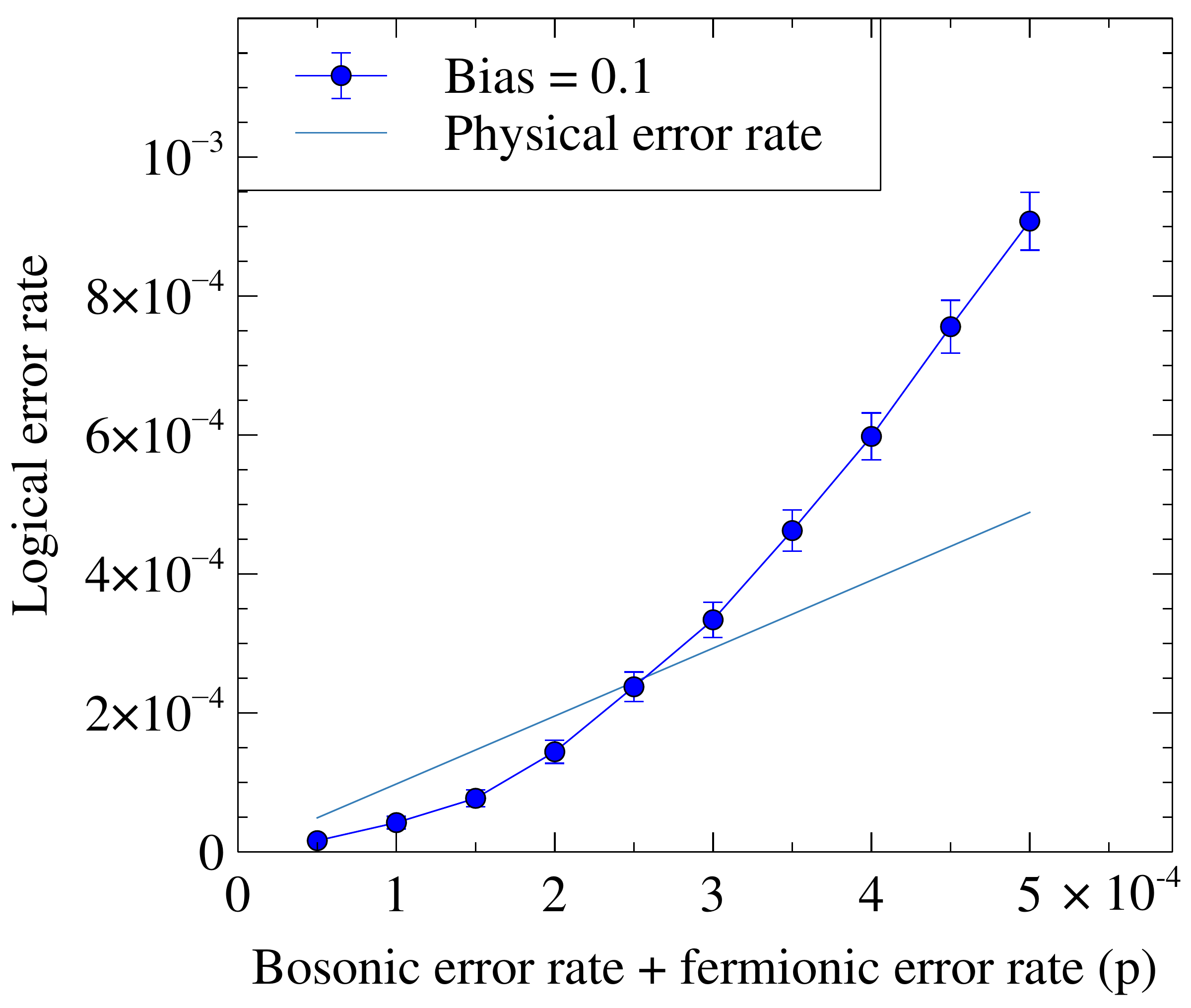}}
		\hspace{0 pt}
		&
		\subfigure[\label{fig:nsf_ft_bias_1}]{\includegraphics[width=0.317\linewidth]{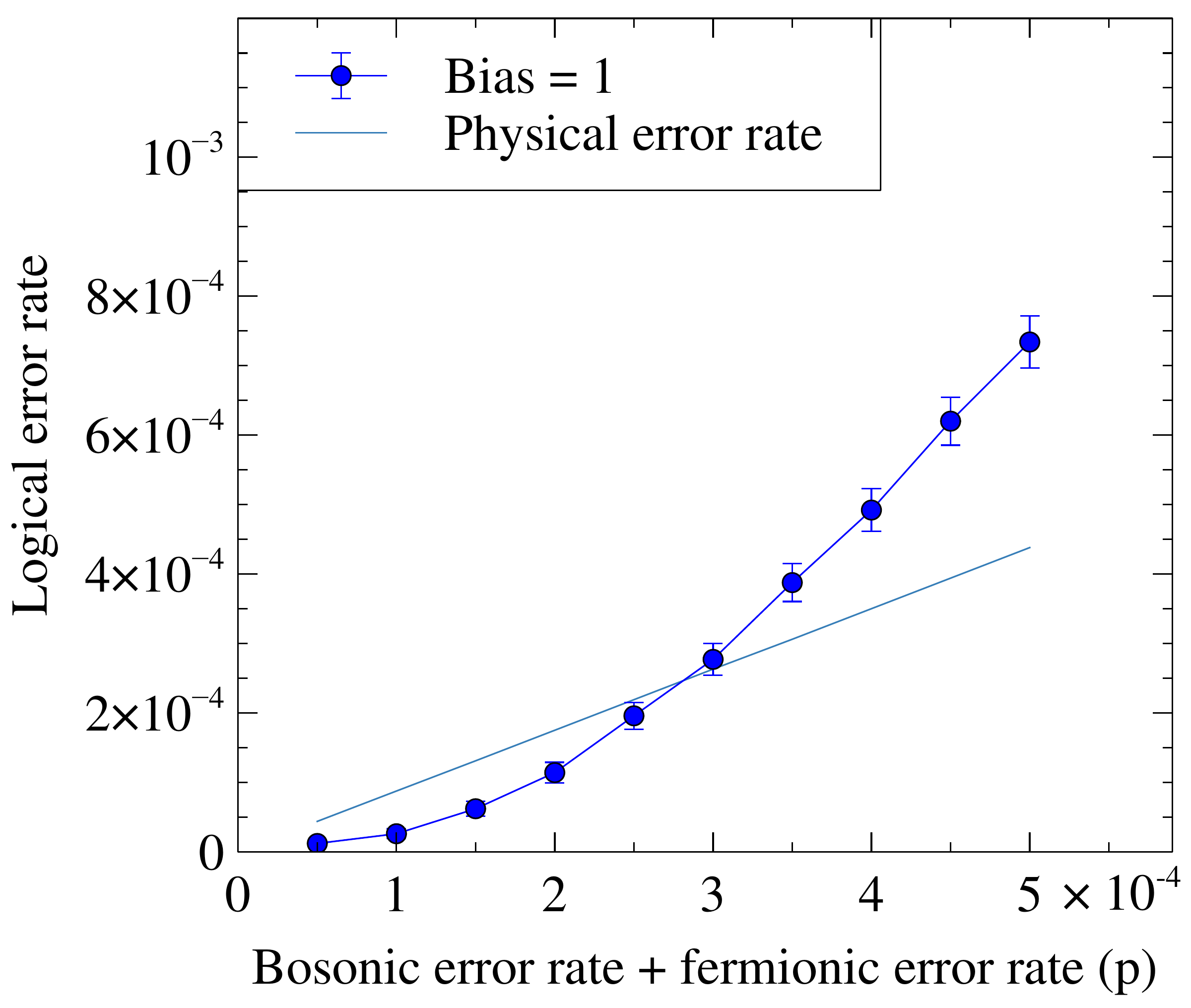}}
		\hspace{0pt}
		&
		\subfigure[\label{fig:nsf_ft_bias_10}]{\includegraphics[width=0.317\linewidth]{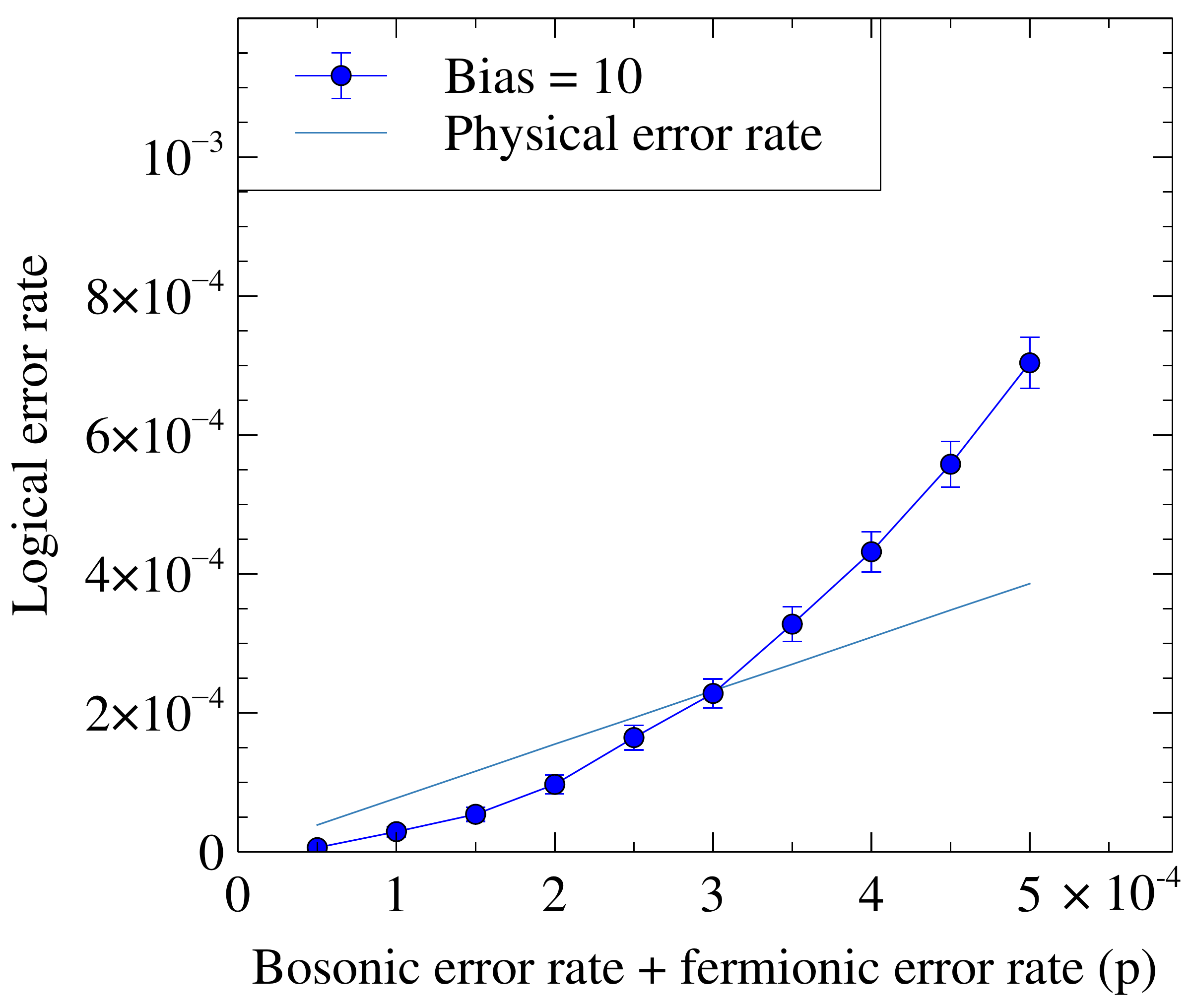}}
		\hspace{0pt}
	\end{tabular}
	\caption{
		The logical error plots for the fault-tolerant implementation of the \( \llbracket 14, 1, 6_f \rrbracket \) code are given in figures (a), (b) and (c) for bias values \(\eta = 0.1, 1, 10 \) respectively. The graph also shows the 95\% confidence intervals. 
	} 
	\label{fig:nsf_ft_threshold}
\end{figure*}

\section{Conclusion}
We have demonstrated that measurements spanning 2 MZMs per tetron are sufficient for fermionic error correction. We have derived a fermionic code family, by placing tetrons in the stabilizer group. Furthermore, we have generalized the construction of similar fermionic codes from conventional stabilizer codes.

The Majorana implementation of an arbitrary stabilizer code may require additional floating topological superconducting links, sometimes known as ``coherent links,'' for syndrome measurement \cite{Karzig2017}. The links have only 2 MZMs, so they are immune to bosonic errors, but can be affected by fermionic errors. 

This is addressed in an upcoming work, where we show that classical error correction can be used to mitigate errors on coherent links. In addition, we also propose link-free architectures for several common stabilizer code families. 
For example, \figref{fig:7_tetron_msmt_scheme} illustrates a Majorana architecture for the \( \llbracket 14, 1, 6_f \rrbracket \) code which does not require coherent links for syndrome measurement.

\begin{figure*}[htb]
	{\includegraphics[width=0.8\linewidth]{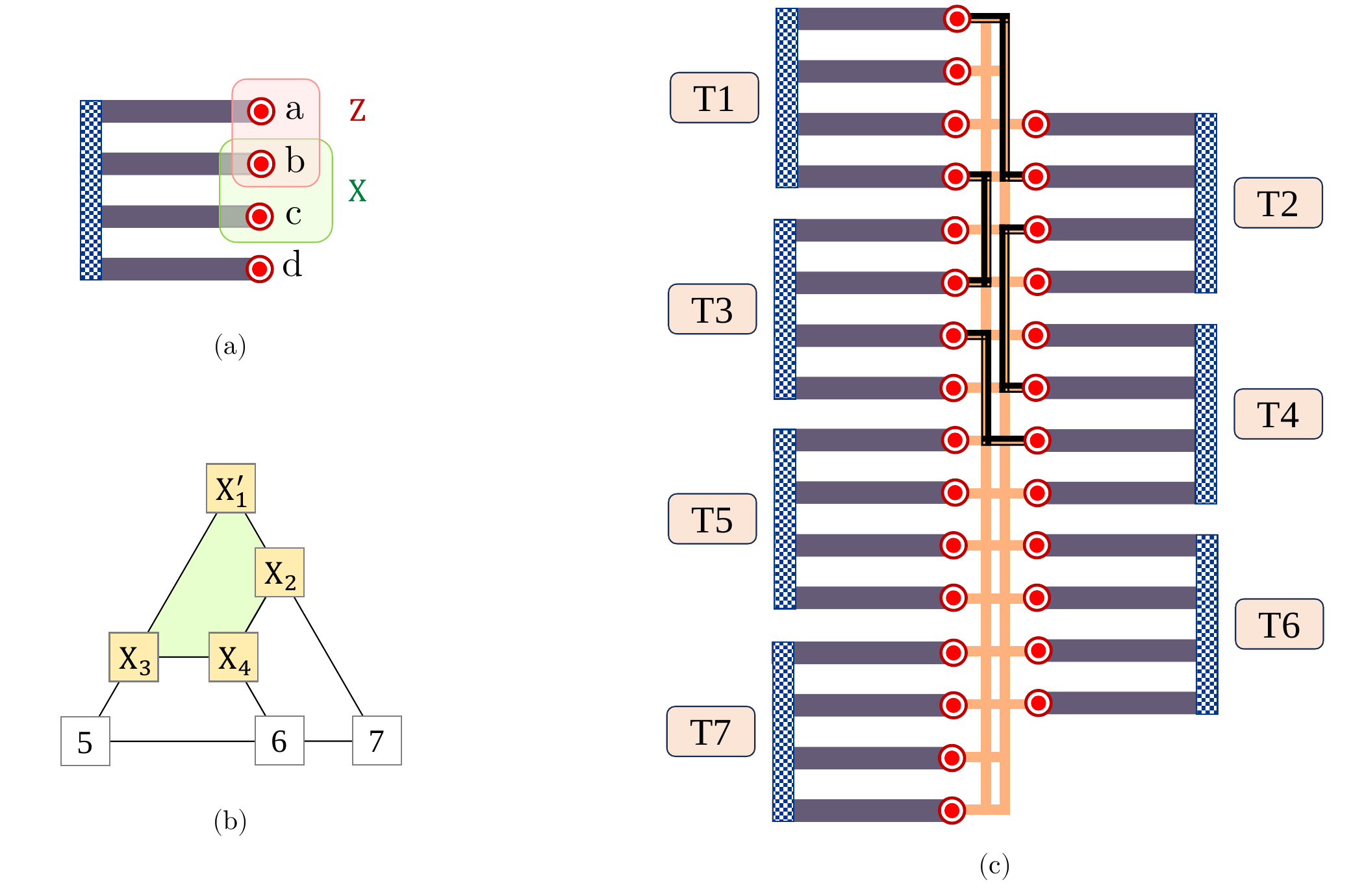}}
	\caption{
		Figure (a) shows a one-sided tetron with MZMs at locations \(a, b, c, d\). Its Pauli operators are \( X = \gamma_b \gamma_c, Z = \gamma_a \gamma_b\).
		Figure (b) highlights one stabilizer of the \( \llbracket 14, 1, 6_f \rrbracket \) code which is being measured. 
		Figure (c) shows the tetron configuration for this stabilizer code. We use doubled black lines to denote quantum dot mediated coupling, which is used to create the measurement loop corresponding to the \(X_1' X_2 X_3 X_4\) stabilizer. 
	} 
	\label{fig:7_tetron_msmt_scheme}
\end{figure*}

\begin{acknowledgments}
This work is supported by NSF grant CCF-1254119, ARO grant W911NF-12-1-0541, and MURI Grant FA9550-18-1-0161.

\end{acknowledgments}

\vspace{451pt}


%

\end{document}

%% file: header.tex
\usepackage{stmaryrd}  
\usepackage{bm}
\usepackage[cal=euler,scr=boondoxo]{mathalfa}

\newlength{\prml}
\usepackage{amstext} 
\usepackage{amsfonts, amssymb, amsmath, amsthm}
\usepackage{latexsym}
\usepackage[tracking=smallcaps]{microtype}      
\usepackage{url}
\usepackage{color}
\usepackage[colorlinks=true,breaklinks, linkcolor=black,citecolor=black,urlcolor=black]{hyperref}

\usepackage{graphicx}
\usepackage[tight, TABBOTCAP]{subfigure}

\newcommand{\ket}[1]{{\left|#1\right\rangle}}

\newcommand{\abs}[1]{{\lvert #1\rvert}} 

\newlength{\commentslength}

\newcommand{\rem}[1]{}

\newtheorem{theorem}{Theorem}

\newtheorem{claim}[theorem]{Claim}

\newfont{\subsubsecfnt}{ptmri8t at 11pt}
\renewcommand{\subparagraph}[1]{\smallskip{\subsubsecfnt #1.}}

\newcommand{\eqnref}[1]{\hyperref[#1]{{(\ref*{#1})}}}
\newcommand{\thmref}[1]{\hyperref[#1]{{Theorem~\ref*{#1}}}}
\newcommand{\lemref}[1]{\hyperref[#1]{{Lemma~\ref*{#1}}}}
\newcommand{\corref}[1]{\hyperref[#1]{{Corollary~\ref*{#1}}}}
\newcommand{\defref}[1]{\hyperref[#1]{{Definition~\ref*{#1}}}}
\newcommand{\secref}[1]{\hyperref[#1]{{Section~\ref*{#1}}}}
\newcommand{\figref}[1]{\hyperref[#1]{{Fig.~\ref*{#1}}}}
\newcommand{\Figref}[1]{\hyperref[#1]{{Figure~\ref*{#1}}}}
\newcommand{\tabref}[1]{\hyperref[#1]{{Table~\ref*{#1}}}}
\newcommand{\remref}[1]{\hyperref[#1]{{Remark~\ref*{#1}}}}
\newcommand{\appref}[1]{\hyperref[#1]{{Appendix~\ref*{#1}}}}
\newcommand{\claimref}[1]{\hyperref[#1]{{Claim~\ref*{#1}}}}
\newcommand{\factref}[1]{\hyperref[#1]{{Fact~\ref*{#1}}}}
\newcommand{\propref}[1]{\hyperref[#1]{{Proposition~\ref*{#1}}}}
\newcommand{\exampleref}[1]{\hyperref[#1]{{Example~\ref*{#1}}}}
\newcommand{\conjref}[1]{\hyperref[#1]{{Conjecture~\ref*{#1}}}}

\allowdisplaybreaks[1]